\documentclass[
a4paper,
aps,
pra,
showpacs,
twocolumn,
nosuperscriptaddress]
{revtex4-2} 

\usepackage[utf8]{inputenc}  
\usepackage[T1]{fontenc}     
\usepackage[american]{babel}  
\usepackage[svgnames,dvipsnames]{xcolor}
\usepackage[colorlinks,citecolor=DarkBlue,linkcolor=DarkBlue,linktocpage,unicode]{hyperref}
\usepackage{amsmath,amssymb,amsthm,bm,mathtools,amsfonts,mathrsfs,bbm,dsfont} 
\usepackage{centernot}
\newtheorem{theorem}{Theorem}
\newtheorem{observation}[theorem]{Observation}

\newcommand\norm[1]{\left\lVert#1\right\rVert}

\newcommand{\tr}{{\mathrm{tr}}}
\newcommand{\va}[1]{\ensuremath{(\Delta#1)^2}}

\newcommand{\ex}[1]{\ensuremath{\langle{#1}\rangle}}

\newcommand{\eins}{\mathbbm{1}}
\newcommand{\swap}{\mathbb{S}}
\renewcommand{\vr}{\ensuremath{\varrho}}
\renewcommand{\vec}[1]{\ensuremath{\boldsymbol{#1}}}
\usepackage{braket}
\usepackage[normalem]{ulem}
\usepackage{comment}
\usepackage{subfigure}

\begin{document}
\title{Metrological usefulness of entanglement and nonlinear Hamiltonians}
\author{Satoya Imai, Augusto Smerzi, and Luca Pezzè}
\affiliation{QSTAR, INO-CNR, and LENS, Largo Enrico Fermi, 2, 50125 Firenze, Italy}
\date{\today}
\begin{abstract}
     A central task in quantum metrology is to exploit quantum correlations to outperform classical sensitivity limits. Metrologically useful entanglement is identified when the quantum Fisher information (QFI) exceeds a separability bound for a given parameter-encoding Hamiltonian. However, so far, only results for linear Hamiltonians are well-established. Here, we characterize metrologically useful entanglement for nonlinear Hamiltonians, presenting separability bounds for collective angular momenta. Also, we provide a general expression for entangled states maximizing the QFI, which can be written as the superposition between the GHZ-like and singlet states. Finally, we compare the metrological usefulness of linear and nonlinear cases, in terms of entanglement detection and random symmetric states.
\end{abstract}
\maketitle
{\it Introduction.---}Certifying and characterizing entanglement is a vibrant area of research in quantum information~\cite{guhne2009entanglement,horodecki2009quantum} and many-body systems~\cite{amico2008entanglement,laflorencie2016quantum,de2018genuine}. This has garnered increased attention, particularly with recent advancements in the experimental manipulation of large quantum systems~\cite{buluta2009simulators,strobel2014fisher,brydges2019probing,omran2019generation,browaeys2020many,ebadi2021quantum,bornet2023scalable,cao2023generation,bluvstein2024logical}. Remarkably, multipartite entanglement is a key resource in several quantum technologies, such as communication~\cite{cleve1999share,hillery1999quantum,lvovsky2009optical}, networks~\cite{kimble2008quantum,simon2017towards,wehner2018quantum,navascues2020genuine,kraft2021quantum,tavakoli2022bell,hansenne2022symmetries}, error correction~\cite{shor1995scheme,gottesman1996class,nielsen2010quantum}, measurement-based computation~\cite{raussendorf2001one,briegel2001persistent}, and metrology~\cite{pezze2009entanglement, hyllus2012fisher, toth2012multipartite, giovannetti2011advances, toth2014quantum, pezze2018quantum}.

In parameter estimation, useful entanglement in a quantum state $\vr$~\cite{pezze2009entanglement} is identified by the quantum Fisher information (QFI)~\cite{helstrom1967minimum,braunstein1994statistical}, denoted as $F_Q(\vr, H)$, which relies on the Hermitian operator $H$. The QFI is related to metrological sensitivity by the quantum Cram\'er-Rao bound $(\Delta \theta)^2_{\rm QCR} = 1/[\upsilon F_Q(\vr,H)]$~\cite{helstrom1967minimum,braunstein1994statistical,giovannetti2011advances, toth2014quantum, pezze2018quantum}. This provides the maximum sensitivity (optimized over all measurement observables and estimators, and saturable after a sufficiently large number, $\upsilon$, of repeated measurements) in the estimation of the parameter $\theta$ encoded in the probe state $\vr$ via the unitary transformation $e^{-i\theta H} \vr e^{i\theta H}$.

Specifically, for a given operator $H$, the bound
\begin{equation} \label{eq:SQL}
    \mathcal{C}_{\rm sep}(H) = \max_{\vr_{\rm sep}} \, F_Q(\vr_{\rm sep}, H),
\end{equation}
gives the maximum QFI over all possible fully separable states $\vr_{\rm sep}$. A certain class of operators $H$ satisfies the inequality $\mathcal{C}_{\rm sep}(H)<\mathcal{C}_{\rm ent}(H)$, where
\begin{equation} \label{eq:HL}
    \mathcal{C}_{\rm ent}(H) = \max_{\vr} \, F_Q(\vr, H),
\end{equation}
gives the maximum QFI over all possible quantum states. This class of operators sets the entanglement criterion:~if $F_Q(\vr,H)>\mathcal{C}_{\rm sep}(H)$, then $\vr$ is entangled. This condition is both necessary and sufficient for \textit{metrologically useful entangled states}~\cite{pezze2009entanglement},~i.e.,~for sensitivity overcoming the standard quantum limit,~$\!(\Delta \theta)^2_{\rm SQL} \!=\! 1/(\upsilon \mathcal{C}_{\rm sep})$,~and ultimately reaching the Heisenberg limit,~$\!(\Delta \theta)^2_{\rm HL} \!=\! 1/(\upsilon \mathcal{C}_{\rm ent})$.

To introduce our setting, let us consider the quantity
\begin{equation}\label{eq:sfraction}
    s(H)
    \equiv \frac{\mathcal{C}_{\rm ent}(H)}{\mathcal{C}_{\rm sep}(H)}
    = \frac{(\Delta \theta)^2_{\rm SQL}}{(\Delta \theta)^2_{\rm HL}}.
\end{equation}
The $s(H)$ represents a gap between $\mathcal{C}_{\rm sep}$ and $\mathcal{C}_{\rm ent}$ and thus quantifies the potential of $H$ to detect metrologically useful entanglement. The condition $s(H)>1$ is both necessary and sufficient for \textit{metrologically useful Hamiltonians}. We call $H_1$ more metrologically useful than $H_2$ if $s(H_1) > s(H_2)$. On the other hand, $H_{\alpha} = \sigma_\alpha \otimes \cdots \otimes \sigma_\alpha$ is an example of not-metrologically useful Hamiltonian, since $s(H_{\alpha}) = 1$, where $\sigma_\alpha$ is Pauli matrix with $\alpha$-axis direction. Also, the case with $H = \eins$ is undefined, since $F_Q(\vr, \eins) = 0$ for any state $\vr$.

\begin{figure}[t]
    \centering
    \includegraphics[width=0.9\columnwidth]{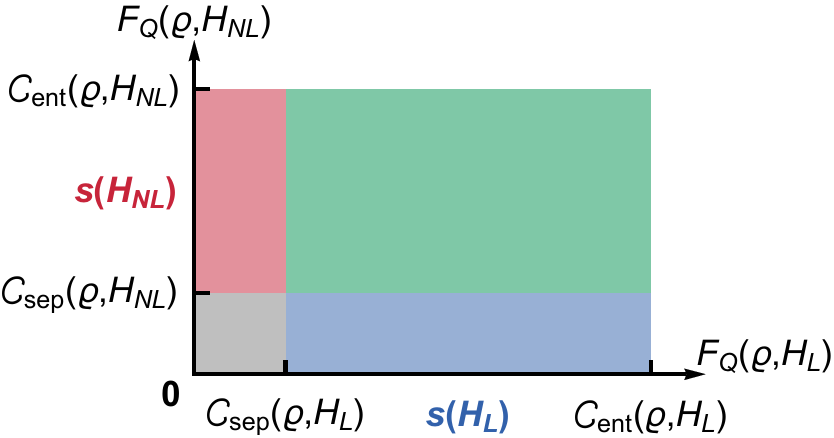}
    \caption{A systematic picture of detection regions for metrologically useful entanglement. The space is coordinated by the quantum Fisher information of linear (horizontal axis) and nonlinear (vertical axis) Hamiltonians.}
    \label{fig:1}
\end{figure}

For linear Hamiltonians, $H_{\rm L} = \sum_{i=1}^N H_i$, where $N$ is the number of particles and $H_i$ are Hamiltonians for individual particles (all $H_i$ are equal), one has $s(H_{\rm L}) = N$~(that will be derived later following Refs.~\cite{giovannetti2006quantum,pezze2009entanglement}). In contrast, for nonlinear Hamiltonians, e.g., $H_{\rm NL} = (\sum_{i=1}^N H_i)^k$, only the scaling behavior is attainable: $s(H_{\rm NL}) = \mathcal{O}(N^{2k})/\mathcal{O}(N^{2k-1}) = \mathcal{O}(N)$, following Refs.~\cite{boixo2007generalized,boixo2008quantumA}, while the exact computation of $\mathcal{C}_{\rm sep}(H_{\rm NL})$ and thus $s(H_{\rm NL})$ is a hard problem. Nonlinear Hamiltonians have been also studied in the context of quantum metrology~\cite{luis2004nonlinear,beltran2005breaking,boixo2007generalized, boixo2008quantum,boixo2008quantumA,choi2008bose,roy2008exponentially,boixo2009quantum,napolitano2011interaction,sewell2014ultrasensitive,beau2017nonlinear} and discussions on the Heisenberg limit~\cite{zwierz2010general, zwierz2012ultimate, gorecki2020pi}.

The goal of this manuscript is to address the following questions:~(i)~whether an ordering relation exists between $s(H_{\rm L})$ and $s(H_{\rm NL})$, therefore determining which of $H_{\rm L}$ and $H_{\rm NL}$ is more metrologically useful;~(ii)~whether there is metrologically useful entanglement detectable by $H_{\rm NL}$ but undetectable by $H_{\rm L}$, or detectable by both operators, unraveling the novel classes of red and green areas shown in Fig.~\ref{fig:1};~(iii)~whether states in the green area typically, or rarely, emerge in the symmetric subspace, as a generalization of the previous result in Ref.~\cite{oszmaniec2016random}.

To proceed, we focus on $H_{\rm NL} \!= \!J_\alpha^k$ for $J_\alpha \!=\! \tfrac{1}{2}\sum_{i=1}^N \sigma_\alpha^{(i)}$, a collective angular momentum along a $\alpha$-axis direction with Pauli matrix $\sigma_\alpha^{(i)}$. First, we outline a general prescription to calculate $\mathcal{C}_{\rm sep}(J_\alpha^k)$ for arbitrary $k$. Especially, we analytically compute $\mathcal{C}_{\rm sep}(J_\alpha^2)$, confirming previous results obtained using numerical methods or restrictions on quantum states~\cite{pezze2016witnessing}. We also derive a general expression for the optimal states reaching $\mathcal{C}_{\rm ent}(J_\alpha^k)$ and find they can always be written as the superposition between the GHZ-like and singlet states for any $k$. Furthermore, we compare metrological usefulness between linear and nonlinear cases based on $s(J_\alpha^k)$. Additionally, we consider specific examples to classify the different regions of Fig.~\ref{fig:1} with noisy states. Finally, we show that most random pure symmetric states achieve $\mathcal{C}_{\rm ent}(J_\alpha^k)$ for all $k$.

\vspace{1em}
{\it Quantum Fisher information.---}The QFI for the unitary transformation $e^{-i\theta H} \vr e^{i\theta H}$ is~\cite{braunstein1994statistical,petz1996monotone}:
\begin{equation}\label{eq:QFI}
    F_Q(\vr, H) = 2\sum_{k,l}
    \frac{(\lambda_k - \lambda_{l})^2}{\lambda_k +\lambda_{l}}
    \lvert \braket{k|H|l} \rvert^2,
\end{equation}
where $\ket{k}$ are eigenvectors of the state $\vr$ with eigenvalues $\lambda_k$ and the sum runs over indices $k,l$ such that $\lambda_k + \lambda_{l} > 0$. The QFI is convex in the state~\cite{fujiwara2001quantum,pezze2009entanglement,toth2014quantum}:~$F_Q\left(\sum_i p_i \ket{\phi_i}\! \bra{\phi_i}, H\right) \leq \sum_i p_i F_Q(\ket{\phi_i}, H)$. This implies that the QFI is always maximized by pure states~\cite{braunstein1994statistical,toth2014quantum,pezze2018quantum}.  For a pure state $\ket{\phi}$, the QFI equals the variance:~$F_Q(\ket{\phi}, H)= 4\va{H}_{\ket{\phi}}$, where $\va{H}_\vr = \tr(\vr H^2) - [\tr(\vr H)]^2$~\cite{braunstein1994statistical}. That is, the bounds $\mathcal{C}_{\rm sep}(H)$ and $\mathcal{C}_{\rm ent}(H)$ in Eqs.~(\ref{eq:SQL})~and~(\ref{eq:HL}) are computed as the maximum variance over pure states.

\vspace{1em}
{\it Computation of $\mathcal{C}_{\rm sep}(H)$.---}An $N$-partite pure state is fully separable if $\ket{\phi_{\rm sep}} = \bigotimes_{i=1}^N \ket{\phi_i}$. In general, a mixed state is fully separable if it can be written as $\vr_{\rm sep} = \sum_k p_k \ket{\phi_{\rm sep}^{(k)}} \! \bra{\phi_{\rm sep}^{(k)}}$, where the $p_k$ form a probability distribution~\cite{horodecki2009quantum,guhne2009entanglement}. For the linear Hamiltonian $H_{\rm L} = \sum_{i=1}^N H_i$, the separability bound $\mathcal{C}_{\rm sep}(H_{\rm L})$ in Eq.~(\ref{eq:SQL}) can be attained only by a symmetric state $\ket{\phi_{\rm sep}^\ast} = \ket{\phi}^{\otimes N}$, i.e., $\mathcal{C}_{\rm sep}(H_{\rm L}) = 4N \max_{\ket{\phi}} \, \va{H_i}_{\ket{\phi}}$~\cite{gessner2016efficient}. This greatly simplifies the calculation, yielding the existing bound $\mathcal{C}_{\rm sep}(J_\alpha) = N$~\cite{pezze2009entanglement}. Note that this derivation relies on the additivity of the QFI~\cite{toth2014quantum}.

For nonlinear Hamiltonians, however, computing the separability bound is more involved as the optimal separable state is not necessarily symmetric, e.g., for $H = J_x^2 + J_y^2, \, {\rm and}\, J_x^2 + J_y^2 + J_z^2$, as we numerically verified. Also, the additivity is not available for nonlinear Hamiltonians.
We summarize our results:
\begin{observation}\label{ob:sepbounds}
    Consider an $N$-qubit Hamiltonian $H_{\rm NL} = J_\alpha^k$.
    (a) For $k=2$, the separability bound $\mathcal{C}_{\rm sep}(J_\alpha^2)$ in Eq.~(\ref{eq:SQL}) can be attained only by a symmetric state. The explicit form of $\mathcal{C}_{\rm sep}(J_\alpha^2)$ for any direction $\alpha$ is given by
    \begin{equation}
        \mathcal{C}_{\rm sep}(J_\alpha^2)
        = \frac{(N-1)^3 N}{2 (2 N-3)},
        \label{eq:sepbound2}
    \end{equation}
    for $N\geq 3$.
    (b) For $k\geq 3$, the separability bound $\mathcal{C}_{\rm sep}(J_\alpha^k)$ may be attained by a symmetric state, supported by numerical evidence up to $N=8$ and $k=5$. In particular, the explicit form of $\mathcal{C}_{\rm sep}(J_\alpha^3)$ for any $\alpha$ may be given by
    \begin{equation}
        \mathcal{C}_{\rm sep}(J_\alpha^3)
        =
        \frac{9 N^5-18 N^4 -b + c_1 + c_2}{216},
        \label{eq:sepbound3}
    \end{equation} 
    where $b = 120 N^3 + 180 N^2 +1020 N$ and the explicit forms of $c_1, c_2$ are given in Appendix~\ref{ap:additional_notes_sepbounds} in the Supplemental Material~\cite{supplmaterial}.
\end{observation}

The proof of Observation~\ref{ob:sepbounds} is given below. The bounds $\mathcal{C}_{\rm sep}(J_\alpha^k)$ can be saturated by $\ket{\phi_{\rm sep}^\ast} \!=\! [\cos(\theta_\ast)\ket{\alpha_+} \!+\! \sin(\theta_\ast) \ket{\alpha_-}]^{\otimes N}$ where the eigenstates $\ket{\alpha_{\pm}}$ with $\pm 1$ eigenvalues of Pauli matrix $\sigma_\alpha$ and $\theta_\ast$ depends on $N$ and $k$. For $k\!=\!2$, $\theta_\ast \!=\! (1/4) \sec^{-1}(3-2N)$ approaches ${\pi}/{8}$ for large $N$ [$\cos({\pi}/{8})\! \approx \! 0.924$]. For $k\!=\! 3$, the explicit form may not be available, but $\cos(\theta_\ast)\! \approx \! 0.953$ for large $N$.

\begin{proof}
    The separability bound $\mathcal{C}_{\rm sep}(J_\alpha^k)$ is found by maximizing $\va{J_\alpha^k}_{\ket{\phi_{\rm sep}}}$ over separable states $\ket{\phi_{\rm sep}} =\bigotimes_{i=1}^N \ket{\phi_i}$. In principle, one can express $\va{J_\alpha^k}_{\ket{\phi_{\rm sep}}}$ as a polynomial: $\va{J_\alpha^k}_{\ket{\phi_{\rm sep}}} = P_k(\alpha_1, \ldots, \alpha_N) \equiv P_k(\vec{\alpha})$ with $\alpha_i = \braket{\phi_i|\sigma_\alpha^{(i)}|\phi_i} \in [-1,1]$ for $1 \leq i \leq N$. 
    
    (a) For $k=2$, we find
    \begin{align}\nonumber
        P_2(\vec{\alpha})
        &=\frac{1}{8}\Bigl[
        N(N-1) + 2(N-2) \sum_{i\neq j}\alpha_i \alpha_j
        \\
        &\quad
        - \sum_{i \neq j} \alpha_i^2 \alpha_j^2
        - 2 \sum_{i \neq j \neq k} \alpha_i^2 \alpha_j \alpha_k
        \Bigr],\label{eq:polyk=2}
    \end{align}
    where we used the expansion of $J_\alpha^2, J_\alpha^4$ in Appendix~\ref{ap:additional_notes_sepbounds}, and $i \! \neq \! j \! \neq \! k$ means $i \! \neq \! j$, $j \! \neq \! k$, and $k\! \neq \! i$. Importantly, since $P_2(\vec{\alpha})$ is symmetric under any exchange of $\alpha_i$ and $\alpha_j$, the fundamental theorem of symmetric polynomials~\cite{macdonald1998symmetric} states that it can be written as a polynomial of $p_{m} \!=\! \sum_{i=1}^N \alpha_i^m$. We apply the method of Lagrange multipliers to maximize $P_2(\vec{\alpha})$ by considering the Lagrangian
    \begin{equation}
        \mathcal{L}(\vec{\alpha}, \kappa_1, \kappa_2, \kappa_3)
        = P_2(\vec{\alpha})
        +\!\! \sum_{m=1,2,3} \! \! \kappa_m \biggl(\sum_{i=1}^N \alpha_i^m - p_m \biggr),
    \end{equation}
    where $\kappa_m$ are Lagrange multipliers. The stationary points of the Lagrangian are given by
    \begin{subequations}
    \begin{align}\label{eq:derivativealpha}
        \frac{\partial \mathcal{L}}{\partial \alpha_x}\!=\!0
        &\Rightarrow 
        \frac{\partial P_k(\vec{\alpha})}{\partial \alpha_x}
        \!+\! \kappa_1 \!+\! 2 \kappa_2 \alpha_x \!+\! 3 \kappa_3 \alpha_x^2 =0,
        \\
        \frac{\partial \mathcal{L}}{\partial \kappa_m}\!=\!0
        &\Rightarrow 
        \sum_{i=1}^N \alpha_i^m = p_m,
        \label{eq:derivativekappa}
    \end{align}
    \end{subequations}
    for $x \in [1,N]$ labeling the particle and $m=1,2,3$. Summing up Eq.~(\ref{eq:derivativealpha}) for all $x$ and using Eq.~(\ref{eq:derivativekappa}) yields 
    \begin{equation}\label{eq:kappa1}
        \kappa_1 = \frac{N p_3 - 3 (\kappa_3 p_2 + p_3) + K_1 p_1 -p_1^3}{N},
    \end{equation}
    where $K_1 = 2 - 2 \kappa_2 + 4 p_2 - N (3 - N + p_2)$. Then substituting Eq.~(\ref{eq:kappa1}) into Eq.~(\ref{eq:derivativealpha}) yields the equation
    \begin{align}\nonumber
        &\!\! 3 \alpha _x^3 -3 p_1 \alpha _x^2
        +(p_1^2-p_2+N-2) \alpha_x
        +2\kappa_2 \Bigl(\alpha_x-\frac{p_1}{N}\Bigr)
        \\
        &\!\!\!\!
        +\! 3\kappa_3 \Bigl(\alpha _x^2-\frac{p_2}{N}\Bigr)
        -\frac{p_1 (N-2-4 p_2)+p_1^3+3 p_3}{N}
        \!=\! 0.\label{eq:equationalphaderivative}
    \end{align}
    By taking Eq.~(\ref{eq:equationalphaderivative}) for different $x \! \neq \! y  \! \neq \! z$ ($y,z\!=\!1,\ldots,N$), one can eliminate $\kappa_2$ and $\kappa_3$ to find the relations
    \begin{subequations}
        \begin{align}\label{eq:equalcond}
        \alpha_x \!=\! \frac{p_1}{N},
        \qquad
        &\alpha_x \!=\! \pm \frac{\sqrt{p_2}}{\sqrt{N}},
        \\
        \label{eq:nonequalcond}
        \alpha _x \!=\! \frac{p_1 \alpha_y \!-\! p_2}{N \alpha_y \!-\! p_1},
        \
        &\alpha_x \!=\!
        \frac{p_1 \alpha_y \alpha_z \!-\! p_2 (\alpha_y \!+\! \alpha_z) \!+\! p_3}
        {N \alpha_y \alpha_z \!-\! p_1 (\alpha_y \!+\! \alpha_z) \!+\! p_2}.
        \end{align}
    \end{subequations}
    
    Accordingly, Eq.~(\ref{eq:equalcond}) leads to the symmetric solution $\vec{\alpha}_{\ast} = (\alpha_{\ast}, \ldots, \alpha_{\ast})$, i.e., all variables are equal. This implies the nontrivial stationary points satisfying Eq.~(\ref{eq:derivativealpha}):
    \begin{equation}
        \alpha_\ast =  \pm \sqrt{\frac{(N-2)}{(2 N-3)}}.
        \label{eq:soultionslphak2}
    \end{equation}
    In Appendix~\ref{ap:additional_notes_sepbounds}, we show that the Hessian matrix of $P_2(\vec{\alpha})$ is negative-semidefinite at the stationary points $\vec{\alpha}_{\rm max} \!=\! (\alpha_{\ast}, \ldots, \alpha_{\ast})$ with $\alpha_\ast$ in Eq.~(\ref{eq:soultionslphak2}), meaning that the maximum of $P_2(\vec{\alpha})$ can be attained by $\vec{\alpha}_{\rm max}$. This yields the bound $\mathcal{C}_{\rm sep}(J_\alpha^2)$ in Eq.~(\ref{eq:sepbound2}) in Observation~\ref{ob:sepbounds}. On the other hand, Eq.~(\ref{eq:nonequalcond}) can lead to the cases where $\vec{\alpha}_{12} \!=\! (\alpha_{r_1}, \ldots, \alpha_{r_1}, \alpha_{r_2}, \ldots, \alpha_{r_2})$ or $\vec{\alpha}_{123} \!=\! (\alpha_{r_1}, \ldots, \alpha_{r_1}, \alpha_{r_2}, \ldots, \alpha_{r_2}, \alpha_{r_3}, \ldots, \alpha_{r_3})$, for $r_i$-times $\alpha_{r_i}$ with $r_i \! \in \! [1,N-1]$ ($i \!=\!1,2,3$). These can yield that $P_2(\vec{\alpha}) \! = \! 0$ attained by $r_1, r_2 \! = \! 1, -1$ or $r_1, r_2, r_3 \! = \! 1, 0, -1$.
    
    (b) For a general $k$, the computation of $P_k(\vec{\alpha})$ is more involved. Nevertheless, it has the following properties:
    (i)~real and non-negative;
    (ii)~symmetric under any exchange of $\alpha_i$ and $\alpha_j$;
    (iii)~non-homogeneous with degree $2k$; and
    (iv)~each $\alpha_i$ has at most a second power.

    To maximize $P_k(\vec{\alpha})$, all first derivatives must be zero. Utilizing property~(iv), one can express the condition as:
    \begin{equation}\label{eq:derivativezero}
        \frac{\partial P_k(\vec{\alpha})}{\partial \alpha_x}=0
        \Rightarrow
        \alpha_x = Q_k (\alpha_1, \ldots, \Hat{\alpha}_x, \ldots, \alpha_N),
    \end{equation}    
    where $Q_k (\alpha_1, \ldots, \Hat{\alpha}_x, \ldots, \alpha_N)$ should be symmetric under any exchange of $\alpha_i$ and $\alpha_j$ for $i,j \neq x$ with the hat symbol indicating omission of $\alpha_x$. Since Eq.~(\ref{eq:derivativezero}) must hold for all $x$, a stationary point exists where all variables are equal, implying a symmetric solution as a necessary condition for a maximum. Although not sufficient for a global maximum, we found numerical evidence for this up to $N = 8$ and $k=5$. For $k=3$, we confirmed that the Hessian matrix is negative-semidefinite at the maximum up to $N=15$ via symbolic calculation with Mathematica. See Appendix~\ref{ap:additional_notes_sepbounds} for further details.
\end{proof}

We have several remarks. First, the separability bound in Eq.~(\ref{eq:sepbound2}) agrees with the one first derived in Ref.~\cite{pezze2016witnessing}, where the maximum QFI over separable states was conjectured to be attained by symmetric states. More recent works~\cite{li2019sensitivity,li2023quantum,bhattacharyya2024even,cieslinski2024exploring} have used the same assumption. Second, the bounds in Eqs.~(\ref{eq:sepbound2})~and~(\ref{eq:sepbound3}) scale as $\mathcal{O}(N^{2k-1})$ for $k=2,3$, as discussed in similar proofs for asymptotic cases~\cite{boixo2007generalized,boixo2008quantumA} and small $N$~\cite{cieslinski2024exploring}. Higher-order cases for $k \geq 4$ may be also computed based on Observation~\ref{ob:sepbounds}. Third, more mathematically, our concern involves determining whether a real, non-negative, non-concave, non-homogeneous symmetric polynomial $f(x_1, \ldots, x_N)$, achieves global extrema at $x_1 = \cdots = x_N$. While this is not generally true, such a symmetric solution typically represents local extrema, according to the Purkiss principle~\cite{waterhouse1983symmetric}. Finally, similar to Observation~\ref{ob:sepbounds}, the closest fully separable state to any symmetric entangled state is also symmetric concerning the geometric measure~\cite{hubener2009geometric}.

\vspace{1em}
{\it Computation of $\mathcal{C}_{\rm ent}(H)$.---}Computing $\mathcal{C}_{\rm ent}(H)$ in Eq.~(\ref{eq:HL}) is much simpler than $\mathcal{C}_{\rm sep}(H)$, using the fact that the variance of a general Hamiltonian $H$ is bounded by $\va{H}_{\vr} \! \leq \! (1/4)(h_{\rm max} \!-\! h_{\rm min})^2$, where $h_{\rm max/min}$ is the max/min eigenvalue of $H$~\cite{textor1978theorem} (for reader convenience, we give this proof in Appendix~\ref{textortheorem_proof_useful}). In quantum metrology, this bound was used first in Ref.~\cite{giovannetti2006quantum} for linear Hamiltonians and in Refs.~\cite{boixo2007generalized,roy2008exponentially} for nonlinear ones. This yields $\mathcal{C}_{\rm ent}(H) \!=\! (h_{\rm max} - h_{\rm min})^2$, i.e., for $H = J_\alpha^k$, we have:
\begin{equation}\label{eq:allbound}
\mathcal{C}_{\rm ent}(J_\alpha^k)
=
\begin{cases}
\frac{N^{2k}}{4^{k-1}},
&\text{odd} \ k \ \text{and} \ \text{all} \ N,
\\
\frac{N^{2k}}{4^{k}},
&\text{even} \ k \ \text{and} \ \text{even} \ N,
\\
\frac{(N^k - 1)^2}{4^{k}},
&\text{even} \ k \ \text{and} \ \text{odd} \ N,
\end{cases}
\end{equation}
where $J_\alpha^{k}$ has $h_{\rm max} = (N/2)^k$, $h_{\rm min} = (-N/2)^k$ for odd $k$, or $h_{\rm min} = 0 \ {\rm or} \ 1/2^k$ for even $k$.

If the max/min eigenvalues of the Hamiltonian $H$ are not degenerate, the state $\ket{\Psi} \equiv (\ket{h_{\rm max}}+ \ket{h_{\rm min}})/\sqrt{2}$, achieves $\mathcal{C}_{\rm ent}(H)$ uniquely, where $\ket{h_{\rm max/min}}$ are the eigenstates with the eigenvalues $h_{\rm max/min}$.  Otherwise, $\ket{\Psi}$ is not uniquely decided. In particular, we have:
\begin{observation}\label{ob:metrologicallymaximallyuseful}    
    Consider an $N$-qubit Hamiltonian $H_{\rm NL} = J_\alpha^k$. The bound $\mathcal{C}_{\rm ent}(J_\alpha^k)$ for any direction $\alpha$ in Eqs.~(\ref{eq:HL})~and~(\ref{eq:allbound}) can be attained by the state
    ~\begin{equation}
    \label{eq:maxusefulstate}
    \ket{\Phi}
    = \sqrt{\lambda_1} \ket{\alpha_+^{N}}
    +\sqrt{\lambda_2}  \ket{\alpha_-^{N}}
    +\sqrt{\lambda_3} \ket{S_N},
\end{equation}
when $\lambda_1 = \lambda_2 = 1/2$ and $\lambda_3 = 0$ for odd $k$, or $\lambda_1 + \lambda_2 = 1/2$ and $\lambda_3 = 1/2$ for even $k$. Here, $\ket{\alpha_{\pm}^{N}} \equiv \ket{\alpha_\pm}^{\otimes N}$ with the eigenstates $\ket{\alpha_\pm}$ with $\pm 1$ eigenvalues of $\sigma_\alpha$ and $\ket{S_N}$ is invariant under any local unitary $U^{\otimes N}$, up to a global phase $\varphi$: $U^{\otimes N}\ket{S_N} = e^{i \varphi}\ket{S_N}$.
\end{observation}
The proof of Observation~\ref{ob:metrologicallymaximallyuseful} is given below. The state $\ket{S_N}$, called the singlet state~\cite{cabello2003solving,cabello2003supersinglets,toth2007optimal,adamson2008detecting,urizar2013macroscopic,bernards2024multiparticle}, is also known as the Werner state in quantum information~\cite{werner1989quantum,vollbrecht2001entanglement,eggeling2001separability}, or ground states of anti-ferromagnetic Heisenberg spin systems in condensed matter~\cite{ueda1996plaquette,miyahara1999exact,eckert2008quantum,keilmann2008dynamical,lubasch2011adiabatic}. In the simplest case $N=2$, there exists only one singlet state: $\ket{S_2} = (\ket{01}-\ket{10})/\sqrt{2}$ (for other examples, see Appendix~\ref{textortheorem_proof_useful}). For any even $N$, singlet states form a linear subspace with the dimension $d(N) = N!/[(N/2)! (N/2 + 1)!]$, referred to as the decoherence-free subspace~\cite{zanardi1997noiseless,lidar1998decoherence,lidar2003decoherence}. 
\begin{proof}
    We note that $\ket{S_N}$ are simultaneous eigenstates of $J_\alpha$ for any direction $\alpha$ with zero eigenvalues: $J_\alpha \ket{S_N} = 0$. This implies the orthogonal relation, $\braket{S_N|\alpha_{\pm}^N} = 0$, since $\ket{\alpha_{\pm}^N}$ are the eigenstates of $J_\alpha$ with eigenvalues $\pm N/2$. Using the normalization condition $\sum_{i=1}^3\lambda_i = 1$ yields
    \begin{equation}\label{eq:QFIMMU}
        F_Q(\ket{\Phi}, J_\alpha^k)
        = \frac{N^{2k}}{4^{k-1}}\!
        \left\{\lambda_1 + \lambda_2
        - [\lambda_1 + (-1)^k \lambda_{2}]^2
        \right\}.
    \end{equation}
    Thus we can confirm that the choice of $\lambda_i$ in Observation~\ref{ob:metrologicallymaximallyuseful} reaches the optimal bound $\mathcal{C}_{\rm ent}(J_\alpha^k)$ for any $\alpha$. 
\end{proof}

For odd $k$, the GHZ state $\ket{\text{GHZ}} = (\ket{0}^{\otimes N} + \ket{1}^{\otimes N})/\sqrt{2}$ attains the bound $\mathcal{C}_{\rm ent}(J_z^k)$. This was first discussed for $k=1$ in Refs.~\cite{giovannetti2006quantum}, and recently extended to odd $k$ in Ref.~\cite{bhattacharyya2024even}. For even $k$, the QFI of GHZ states vanishes, because both $\ket{0}^{\otimes N}$ and $\ket{1}^{\otimes N}$ become eigenstates of $J_z^k$ with eigenvalues $(N/2)^k$. This shows that useful entanglement in linear metrology may not be necessarily useful in nonlinear cases. On the other hand, for even $k$, a superposition between the singlet state and other states, i.e., $\ket{\Phi(\lambda)} = \sqrt{\lambda}\ket{0}^{\otimes N}+\sqrt{{1}/{2} - \lambda}\ket{1}^{\otimes N} + \sqrt{{1}/{2}} \ket{S_N}$, attains the bound $\mathcal{C}_{\rm ent}(J_\alpha^k)$, for any $\lambda \in [0,1/2]$. In particular, for $\lambda_1 = \lambda_2 = 1/4$, the state $\ket{\Phi(1/4)} = \tfrac{1}{\sqrt{2}} \left(\ket{\text{GHZ}} + \ket{S_N}\right)$ has $F_Q(\ket{\Phi(1/4)}) = \left(\frac{1}{2} \delta_{k, {\text{odd}}} +  \delta_{k, {\text{even}}} \right) \mathcal{C}_{\rm ent}(J_z^k)$. This illustrates an interesting case of very useful entanglement in both linear and nonlinear metrology.

\vspace{1em}
{\it Computation of $s(H)$.---}As noted in the introduction, any linear Hamiltonian $H_{\rm L} = \sum_{i=1}^N H_i$ (all $H_i$ are equal) has $s(H_{\rm L}) = N$. This follows from the relation $h_{\rm max/min} = N \mathfrak{h}_{\rm max/min}$ for $\mathfrak{h}_{\rm max/min}$ being the max/min eigenvalues of $H_i$ and then $\mathcal{C}_{\rm sep}(H_{\rm L})= N(\mathfrak{h}_{\max}- \mathfrak{h}_{\rm min})^2$ and $\mathcal{C}_{\rm ent}(H_{\rm L}) = N^2(\mathfrak{h}_{\rm max}- \mathfrak{h}_{\rm min})^2$~\cite{giovannetti2006quantum,pezze2009entanglement}. A key aspect of our work is to obtain the exact values of $\mathcal{C}_{\rm sep}(J_\alpha^k)$ and $\mathcal{C}_{\rm ent}(J_\alpha^k)$. This is necessary to discuss the metrological usefulness of $J_\alpha^k$. We summarize our results:
\begin{observation}\label{ob:sfractionValue}    
    Consider an $N$-qubit Hamiltonian $H_{\rm NL} = J_\alpha^k$. Letting $s_k \equiv s(J_\alpha^k)$ in Eq.~(\ref{eq:sfraction}), we have that $s_2 < s_1 < s_3$ for $3 \leq N \leq 6$ and $s_2 < s_3 < s_1$ for $N \geq 7$, where $s_1 = N$ and the explicit forms of $s_2, s_3$ are given in Appendix~\ref{ap:additional_notes_sfraction}.
\end{observation}
Observation~\ref{ob:sfractionValue} implies that for large $N$, the linear Hamiltonian $J_\alpha$ is more metrologically useful than the nonlinear ones $J_\alpha^2$ and $J_\alpha^3$.  This addresses the question~(i)~posed in the introduction. Interestingly, it also disproves the presence of a hierarchical order $s_k > s_{k+1}$. On the other hand, we collect numerical evidence that another hierarchical order $s_{k} > s_{k+2}$ exists for a large $N$, for details see Appendix~\ref{ap:additional_notes_sfraction}. We leave this as a conjecture for further research. Moreover, in Appendix~\ref{ap:complicatedHam}, we discuss $\mathcal{C}_{\rm sep}(H_{\alpha, \beta})$ and $s(H_{\alpha, \beta})$ for the class of Hamiltonians $H_{\alpha, \beta} = \mu J_\alpha + \nu J_\beta^2$ to compare their behaviors for different $\mu, \nu$.

\begin{figure}[t]
     \centering
     \includegraphics[width=0.98\columnwidth]{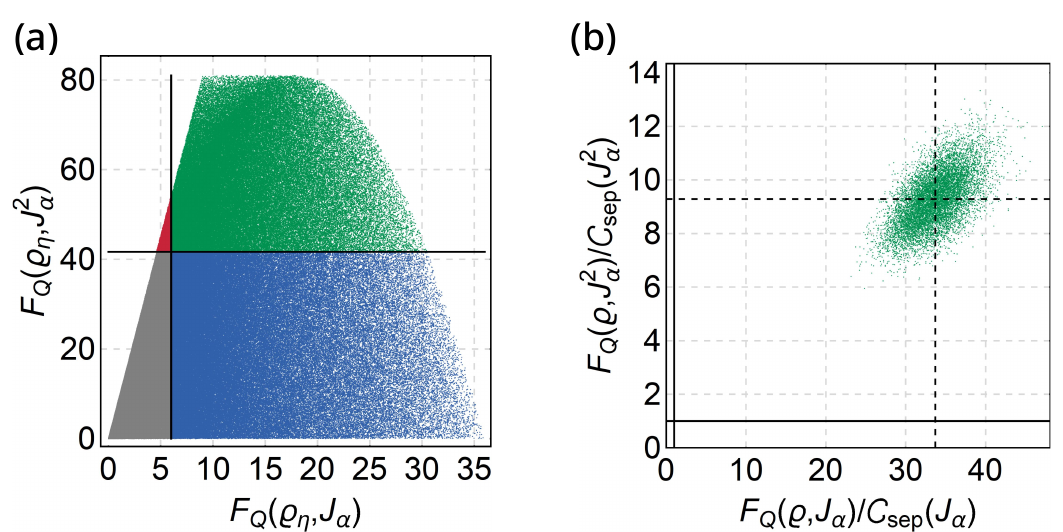}
    \caption{(a) Classification of metrologically useful entanglement for the mixed state $\vr_{\eta}$ in Eq.~(\ref{eq:testmix}) with $N=6$, in the space coordinated by $F_Q(\vr_{\eta}, J_\alpha^k)$ for $k=1,2$. This figure is created from a random sample of $10^6$ points $\lambda_1, \lambda_2, \eta \in [0,1]$ with~$0 \! \leq \!  \lambda_1   +   \lambda_2 \! \leq \! 1$ in Eq.~(\ref{eq:maxusefulstate}).
    (b) Plot of $10^4$ sampled random pure symmetric states for $N=100$ in the space coordinated by $F_Q(\vr, J_\alpha^k)/\mathcal{C}_{\rm sep}(J_\alpha^k)$ for $k=1,2$. Dashed lines are $\Bar{\mathcal{F}}_Q(J_\alpha^k)/\mathcal{C}_{\rm sep}(J_\alpha^k)$ discussed in Observation~\ref{ob:tfractionaverageValue}. In both panels, vertical and horizontal solid lines are respectively $\mathcal{C}_{\rm sep}(J_\alpha) = N$~\cite{pezze2009entanglement} and $\mathcal{C}_{\rm sep}(J_\alpha^2)$ given in Observation~\ref{ob:sepbounds}.
    }
    \label{fig:2}
\end{figure}

\vspace{1em}
{\it Examples.---}We test our bound $\mathcal{C}_{\rm sep}(J_\alpha^2)$ in Observation~\ref{ob:sepbounds} with the noisy-mixed state of $\ket{\Phi}$ in Eq.~(\ref{eq:maxusefulstate}):
\begin{equation}\label{eq:testmix}
    \vr_{\eta} = \eta \ket{\Phi}\!\bra{\Phi} + \frac{1-\eta}{2^N}\eins.
\end{equation}
This state has $F_Q(\vr_{\eta}, J_\alpha^k)\! =\! \tfrac{\eta^2 2^{N-1}}{1+ \eta(2^{N-1}-1)} F_Q(\ket{\Phi}, J_\alpha^k)$ with $F_Q(\ket{\Phi}, J_\alpha^k)$ given in Eq.~(\ref{eq:QFIMMU}).

In Fig.~\ref{fig:2}~(a), we classify useful entanglement in the space by $F_Q(\vr_{\eta}, J_\alpha)$ and $F_Q(\vr_{\eta}, J_\alpha^2)$, with the previous bound $\mathcal{C}_{\rm sep}(J_\alpha) = N$~\cite{pezze2009entanglement} and our bound $\mathcal{C}_{\rm sep}(J_\alpha^2)$ in Eq.~(\ref{eq:sepbound2}). This reveals useful entanglement detectable by our bound but not detectable by the previous one, or detectable by both, thus addressing the question~(ii)~in the introduction.  In Appendix~\ref{ap:spinsqueezing}, we provide more details on the analysis of $\vr_{\eta}$ and compare with entanglement criteria based on $\ex{J_\alpha}$ and $\ex{J_\alpha^2}$~\cite{toth2007optimal,toth2009spin}. There we show that the optimal entangled state $\ket{\Phi(\lambda)}$ is spin-squeezed for any $\lambda \in [0,1/2]$, although the GHZ state is not spin-squeezed.

\vspace{1em}
{\it Average QFI of random symmetric states.---}Finally we discuss the metrological usefulness based on the average QFI of random pure states in the symmetric subspace:
\begin{equation}\label{eq:averageQFIbosonic}
    \Bar{\mathcal{F}}_Q(H)
    \equiv \int d\psi_{\rm sym} \, F_Q(\ket{\psi_{\rm sym}}, H),
\end{equation}
where symmetric states $\ket{\psi_{\rm sym}}$ are chosen via the Haar distribution. In Ref.~\cite{oszmaniec2016random}, the average QFI was shown to exceed the separability bound for linear Hamiltonians: $\Bar{\mathcal{F}}_Q(H_{\rm L}) \geq \mathcal{C}_{\rm sep}(H_{\rm L})$. We summarize our results:
\begin{observation}\label{ob:tfractionaverageValue}    
Consider an $N$-qubit Hamiltonian $H_{\rm NL} = J_\alpha^k$. For large $N$, we obtain that $\Bar{\mathcal{F}}_Q(J_\alpha^k) \propto ({1}/{k}) \mathcal{C}_{\rm ent}(J_\alpha^k)$ (with analytical expression valid for any $k$ and $N$ reported in Appendix~\ref{ap:averageQFI}), leading to the scaling behavior $\mathcal{O}(N^{2k}/k)$. Also, defining the quantity $t_k \equiv \Bar{\mathcal{F}}_Q(J_\alpha^k)/\mathcal{C}_{\rm sep}(J_\alpha^k)$, we have that $t_2 < t_3 < t_1$ for $N \geq 4$.
\end{observation}
The proof of Observation~\ref{ob:tfractionaverageValue} and the explicit forms of $t_k$ are given in Appendix~\ref{ap:averageQFI}. Observation~\ref{ob:tfractionaverageValue} shows that most random pure symmetric states can not only be useful in nonlinear metrology, but also reach the optimal bound $\mathcal{C}_{\rm ent}(J_\alpha^k)$ with a slowing-down term~$\propto 1/k$, thus addressing the question~(iii)~in the introduction. 
In Appendix~\ref{ap:averageQFI}, we demonstrate that the typical value of the QFI can go beyond the separability bounds for large $N$ by computing the concentration inequality.
Moreover, similar to Observation~\ref{ob:sfractionValue}, $J_\alpha$ is \textit{typically} more metrologically useful than $J_\alpha^2$ and $J_\alpha^3$ over random symmetric states. In Fig.~\ref{fig:2}~(b), we plot random pure symmetric states of $N=100$ qubits in the space coordinated by $F_Q(\vr, J_\alpha^k)/\mathcal{C}_{\rm sep}(J_\alpha^k)$.

\vspace{1em}
{\it Conclusion.---}We present the conditions of metrological usefulness of quantum states and Hamiltonians: these are $F_Q(\vr, H)>\mathcal{C}_{\rm sep}(H)$ and $s(H)>1$, respectively. 
We have provided the separability bounds $\mathcal{C}_{\rm sep}(J_\alpha^k)$ and the form of the optimal state achieving the bounds $\mathcal{C}_{\rm ent}(J_\alpha^k)$. Also, we found the ordering relation of the metrological usefulness $s(J_\alpha^k)$. We applied our separability bounds to characterize entanglement and computed the average QFI of random symmetric states. We finally emphasize that nonlinear Hamiltonians are relevant in several experimental platforms~\cite{pezze2018quantum,browaeys2020many,demille2024quantum}.

There are several directions for further research.  First, since the computation of $\mathcal{C}_{\rm sep}(H_{\rm NL})$ is hard in general, it would be valuable to find the general class of Hamiltonians (that includes $H_{\rm NL} = J_\alpha^k$) where the bound is attained by symmetric states. Second, extending our results to higher dimensions would be also interesting. Third, our results encourage further development of the activation of metrological usefulness~\cite{toth2020activating,trenyi2024activation} or multiparameter metrology~\cite{gessner2018sensitivity,albarelli2020perspective,demkowicz2020multi}.
Finally, while linear Hamiltonians have been used in various domains beyond quantum metrology, such as quantum coherence~\cite{streltsov2017colloquium,marvian2020coherence}, resource theory~\cite{tan2021fisher,marvian2022operational,yamaguchi2023beyond}, quantum battery~~\cite{julia2020bounds,campaioli2024colloquium}, and Zeno dynamics or state's indistinguishability~\cite{smerzi2012zeno}, one may investigate these possibilities also for nonlinear Hamiltonians.

\vspace{1em}
{\it Acknowledgments.---}
We would like to thank Francesco Albarelli, Stefano Gherardini, Vittorio Giovannetti, Otfried G\"uhne, Ties Ohst, Hai-Long Shi, G\'eza T\'oth, and Benjamin Yadin for discussions. This work has received funding under Horizon Europe programme HORIZON-CL4-2022-QUANTUM-02-SGA via the project 101113690 (PASQuanS2.1) and by the European Commission through the H2020 QuantERA ERA-NET Cofund in Quantum Technologies projects ``SQUEIS''.

\onecolumngrid
\appendix
\addtocounter{theorem}{-1}
\section{Additional notes on Observation~\ref{ob:sepbounds}}\label{ap:additional_notes_sepbounds}
Here we first compute the Hessian matrix of the polynomial $P_2(\vec{\alpha})$ in Eq.~(\ref{eq:polyk=2}) in the main text. Second, we provide the explicit form of Eq.~(\ref{eq:sepbound3}) in Observation~\ref{ob:sepbounds} in the main text and continue the explanation of \textit{proof}. Finally, we give the explicit form of $J_\alpha^k$ for $k \in [1,6]$.
\begin{itemize}
    \item 
    Let us show that the Hessian matrix of the polynomial $P_2(\vec{\alpha})$ is negative-semidefinite (i.e., all eigenvalues are non-positive) at the stationary points $\vec{\alpha}_{\rm max} = (\alpha_{\ast}, \ldots, \alpha_{\ast})$ with $\alpha_\ast = \pm \sqrt{\tfrac{(N-2)}{(2 N-3)}}$ in Eq.~(\ref{eq:soultionslphak2}) in the main text. Let us begin by denoting the Hessian matrix as $\mathcal{H}(\vec{\alpha}) = (\mathcal{H}_{xy})$ with the $(x,y)$-element $\mathcal{H}_{xy} = \tfrac{\partial^2 P_2(\vec{\alpha})}{\partial \alpha_x \partial \alpha_y}$. A straightforward calculation leads to
    \begin{equation} \label{eq:twoderivative}
        \mathcal{H}_{xy}
        =
        \frac{1}{2}
        \begin{cases}
        - \sum_{i \neq x} \alpha_i^2
        - \sum_{i \neq j \neq x} \alpha_i \alpha_j,
        \ &x = y,
        \\ 
        N-2 - 2\alpha_x \alpha_{y}
        - \sum_{i \neq x \neq y}\alpha_i
        (2\alpha_x + \alpha_i + 2\alpha_y),
        \ &x \neq y,
        \end{cases}
    \end{equation}
    where we used
    \begin{equation}
       \frac{\partial P_2(\vec{\alpha})}{\partial \alpha_x}
        =\frac{1}{2}\Bigl[
        (N-2) \sum_{i\neq x}\alpha_i
        - \sum_{i \neq x} \alpha_x \alpha_i^2
        - \sum_{i \neq j \neq x} (\alpha_x \alpha_i \alpha_j + \alpha_i^2 \alpha_j)
        \Bigr].
    \end{equation}   
    At the stationary points $\vec{\alpha}_{\rm max} = (\alpha_{\ast}, \ldots, \alpha_{\ast})$, the $(x,y)$-element in Eq.~(\ref{eq:twoderivative}) is given by
    \begin{equation} \label{eq:hessianalpha}
        \mathcal{H}_{xy}
        =
        \frac{1}{2}
        \begin{cases}
        - (N-1)^2 \alpha_\ast^2,
        \ &x = y,
        \\
        (N-2)(1-5\alpha_\ast^2) - 2\alpha_\ast^2
        \ &x \neq y.
        \end{cases}
    \end{equation}
    Inserting the value $\alpha_\ast = \pm \sqrt{\tfrac{(N-2)}{(2 N-3)}}$ into Eq.~(\ref{eq:hessianalpha}), we obtain 
    \begin{equation}\label{eq:twoderivativeequal}
        \mathcal{H}(\vec{\alpha}_{\rm max}) = - (q_N -  q_N^\prime) \eins_{N} - q_N^\prime {\mathbbm{J}}_N,
    \end{equation}
    where ${\mathbbm{J}}_N$ is a $N \times N$ matrix where all the elements are equal to one and
    \begin{equation}
        q_N = \frac{(N-2) (N-1)^2}{2(2N-3)},
        \qquad
        q_N^\prime = \frac{(N-2) (3N-5)}{2(2N-3)}.
    \end{equation}
    Since $\eins_{N}, {\mathbbm{J}}_N$ are positive-semidefinite and $q_N \geq q_N^\prime$ for $N \geq 3$, the Hessian matrix $\mathcal{H}(\vec{\alpha}_{\rm max})$ in Eq.~(\ref{eq:twoderivativeequal}) is negative-semidefinite.
    
    \item 
    Let us provide the explicit form of Eq.~(\ref{eq:sepbound3}) in Observation~\ref{ob:sepbounds} in the main text:
    \begin{subequations}
    \begin{align}\label{eq:sepbound3ap}
    \mathcal{C}_{\rm sep}(J_\alpha^3)
    &= \frac{1}{216}
    \left(
    9 N^5-18 N^4-120 N^3 -180 N^2
    -1020 N + c_1 + c_2 \right),
    \\
    c_1&= 
    \frac{380 (164-71 N)}{3 (N-5) N+20}
    + \frac{12800 (N-1)}{[3 (N-5) N+20]^2}
    -3084,
    \\ 
    c_2
    &=3 \sqrt{\frac{N^2 [N(N c_3 -1440) + 480]^3}
    {(N-2) (N-1) [3 (N-5) N+20]^4}},
    \\
    c_3&= 
    3 N \{N [3 (N-9) N+128]-360\}+1720.
\end{align}
\end{subequations}

    Also, we can continue the explanation of \textit{proof} of Observation~\ref{ob:sepbounds} in the main text as follows: For the case of $k=3$, let us assume that a pure symmetric fully separable state $\ket{\phi_{\rm sep}^\ast}$ can attain the maximum value. Inserting the expanded forms of $J_\alpha^3, J_\alpha^6$ (that will be given in below) into the variance $\va{J_\alpha^3}_{\ket{\phi_{\rm sep}^\ast}}$, we can derive
\begin{align} \nonumber
\va{J_\alpha^3}_{\ket{\phi_{\rm sep}^\ast}}
&= \ex{J_\alpha^6}_{\ket{\phi_{\rm sep}^\ast}}
- \ex{J_\alpha^3}_{\ket{\phi_{\rm sep}^\ast}}^2
\\ \nonumber
&= \frac{N(1-\alpha^2)}{16} \bigl\{
15 N^2 -30 N+16
+3 \alpha^4 (N-2) (N-1) [3 (N-5) N+20]
\\
&\quad
+12 \alpha^2 (N-2) (N-1) (3 N-5)\bigr\},
\end{align}
where we used the facts that $\braket{\phi_{\rm sep}^\ast|\sigma_\alpha^{(i_1)} \sigma_\alpha^{(i_2)} \cdots \sigma_\alpha^{(i_m)} |\phi_{\rm sep}^\ast} = \alpha^m$ for any $m \in [1,N]$, where the Bloch coefficient $\alpha = \braket{\phi|\sigma_\alpha|\phi}$ has the same value $\alpha \in [-1,1]$ for all subsystems. Performing the maximization over $\alpha$, we can arrive at the bound $\mathcal{C}_{\rm sep}(J_\alpha^3)$ in Eq.~(\ref{eq:sepbound3}) in the main text or in Eq.~(\ref{eq:sepbound3ap}). 

    \item 
    For reader convenience, let us provide the explicit form of $J_\alpha^k$ for $k \in [1,6]$, where $J_\alpha = (1/2)\sum_{i=1}^N \sigma_\alpha^{(i)}$. Using the property that $\sigma_\alpha^m = \sigma_\alpha \delta_{m,{\rm odd}} + \eins \delta_{m,{\rm even}}$, we have:
\begin{align}
    J_\alpha^2
    &=\frac{1}{4}\sum_{i,j=1}^N \sigma_\alpha^{(i)}\sigma_\alpha^{(j)}
    =\frac{1}{4}\Biggl\{
    \sum_{i=1}^N (\sigma_\alpha^{(i)})^2
    + \sum_{i\neq j}^N \sigma_\alpha^{(i)} \sigma_\alpha^{(j)}
    \Biggr\}
    = \frac{1}{4}\Biggl\{
    N\eins + \sum_{i\neq j}^N \sigma_\alpha^{(i)} \sigma_\alpha^{(j)}
    \Biggr\},
    \\
    J_\alpha^3 \nonumber
    &=\frac{1}{8}\sum_{i,j,k=1}^N \sigma_\alpha^{(i)}\sigma_\alpha^{(j)}\sigma_\alpha^{(k)}
    =\frac{1}{8}\Biggl\{
    \sum_{i=1}^N (\sigma_\alpha^{(i)})^3
    + 3 \sum_{i\neq j}^N (\sigma_\alpha^{(i)})^2 \sigma_\alpha^{(j)}
    + \sum_{i\neq j \neq k}^N \sigma_\alpha^{(i)}\sigma_\alpha^{(j)}\sigma_\alpha^{(k)}
    \Biggr\}
    \\ \nonumber
    &= \frac{1}{8}\Biggl\{
    \sum_{i=1}^N \sigma_\alpha^{(i)}
    + 3(N-1) \sum_{j=1}^N \sigma_\alpha^{(j)}
    + \sum_{i\neq j \neq k}^N \sigma_\alpha^{(i)}\sigma_\alpha^{(j)}\sigma_\alpha^{(k)}
    \Biggr\}
    \\
    &= \frac{1}{8}\Biggl\{
    (3N-2) \sum_{i=1}^N \sigma_\alpha^{(i)}
    + \sum_{i\neq j \neq k}^N \sigma_\alpha^{(i)}\sigma_\alpha^{(j)}\sigma_\alpha^{(k)}
    \Biggr\},
\end{align}
where $i\neq j \neq k$ means that $i \neq j$, $j \neq k$, and $k \neq i$ (we use similar expressions in the following). Expanding multiple summations similarly, we can obtain:
\begin{align}
    J_\alpha^4
    &=\frac{1}{16}\Biggl\{
    (3N-2)N \eins
    + 2(3N-4) \sum_{i\neq j}^N \sigma_\alpha^{(i)} \sigma_\alpha^{(j)}
    + \sum_{i\neq j \neq k \neq l}^N
    \sigma_\alpha^{(i)}\sigma_\alpha^{(j)}\sigma_\alpha^{(k)}\sigma_\alpha^{(l)}
    \Biggr\},
    \\
    J_\alpha^5
    &=\frac{1}{32}\Biggl\{
    [15N(N-2)+16]\sum_{i=1}^N \sigma_\alpha^{(i)}
    + 10(N-2) \sum_{i\neq j \neq k}^N
    \sigma_\alpha^{(i)} \sigma_\alpha^{(j)} \sigma_\alpha^{(k)}
    + \sum_{i\neq j \neq k \neq l \neq m}^N
    \sigma_\alpha^{(i)}\sigma_\alpha^{(j)}\sigma_\alpha^{(k)}
    \sigma_\alpha^{(l)}\sigma_\alpha^{(m)}
    \Biggr\},
    \\
    J_\alpha^6
    \nonumber
    &=\frac{1}{64}\Biggl\{
    N[15N(N-2)+16] \eins
    + [15 N (3N-10)+136] \sum_{i\neq j}^N \sigma_\alpha^{(i)} \sigma_\alpha^{(j)}
    + 5 (3N-8) \sum_{i\neq j \neq k \neq l}^N \sigma_\alpha^{(i)} \sigma_\alpha^{(j)}\sigma_\alpha^{(k)}
    \sigma_\alpha^{(l)}
    \\
    &\quad
    + \sum_{i\neq j \neq k \neq l \neq m \neq n}^N
    \sigma_\alpha^{(i)}\sigma_\alpha^{(j)}\sigma_\alpha^{(k)}
    \sigma_\alpha^{(l)}\sigma_\alpha^{(m)}\sigma_\alpha^{(n)}
    \Biggr\}.
\end{align}
It might be useful to note that this derivation is related to the so-called partition function in number theory and the so-called Young diagrams in combinatorics.
\end{itemize}
\section{Additional notes on Observation~\ref{ob:metrologicallymaximallyuseful}}\label{textortheorem_proof_useful}
\subsection{Variance's upper bound}
In Ref.~\cite{textor1978theorem}, the variance of a general Hamiltonian $H$ for any state $\vr$ has the upper bound
\begin{equation}
    \va{H}_{\vr} \leq \frac{1}{4}(h_{\rm max} - h_{\rm min})^2,
\end{equation}
where $h_{\rm max/min}$ are the max/min eigenvalues of $H$. For reader convenience, we here give the proof of the above inequality, following the description of Ref.~\cite{imai2023randomized}. We remark that the following proof is similar to the derivation of the Popoviciu inequality~\cite{popoviciu1935equations}, as well as the Bhatia–Davis inequality~\cite{bhatia2000better,sharma2008some}. 
    \begin{proof}
    We note that $h_{\text{min}} \eins \leq H \leq h_{\text{max}} \eins$. This leads to the following inequality:
    \begin{equation}
        0 \leq \ex{(h_{\text{max}} \eins - H) (H - h_{\text{min}} \eins)}_\vr
        \iff
        \ex{H^2}_\vr \leq \ex{H}_\vr (h_{\text{max}} + h_{\text{min}}) - h_{\text{max}} h_{\text{min}}.
    \end{equation}
    Substituting the right-hand-side in this inequality to the variance $\va{H}_\vr = \ex{H^2}_\vr - \ex{H}_\vr^2$ yields
    \begin{equation}
        \va{H}_\vr \leq
        (h_{\text{max}} -\ex{H}_\vr)
        (\ex{H}_\vr - h_{\text{min}}).
    \end{equation}
    Finally, by applying the inequality of arithmetic and geometric means $xy \leq [(x+y)/2]^2$, we can complete the proof.
    \end{proof} 

\subsection{Singlet states}
The singlet states $\ket{S_N}$ are invariant under any local unitary $U^{\otimes N}$, up to a global phase $\varphi$: 
\begin{equation}
    U^{\otimes N}\ket{S_N} = e^{i \varphi}\ket{S_N},
\end{equation}
or equivalently, $U^{\otimes N}\ket{S_N}\! \bra{S_N} (U^{\dagger})^{\otimes N} = \ket{S_N}\! \bra{S_N}$. For any even $N$, singlet states form a linear subspace with the dimension $d(N) = N!/[(N/2)! (N/2 + 1)!]$. 

For $N=2$, there exists only one singlet state, i.e., $d(2)=1$: $\ket{S_2} = (\ket{01}-\ket{10})/\sqrt{2}$. For $N=4$, there exist two singlet states, i.e., $d(2)=2$.
They form a two-dimensional subspace spanned by
$\ket{S_4^{(1)}} = \ket{S_2} \otimes \ket{S_2}$ and 
\begin{equation}
    \ket{S_4^{(2)}} =
    \frac{1}{2\sqrt{3}} \left( 2\ket{0011} - \ket{0101} - \ket{0110} - \ket{1001} - \ket{1010} + 2\ket{1100} \right).
\end{equation}
In Refs.~\cite{cabello2003solving,cabello2003supersinglets}, the family of $N$-qubit singlet states can be written as
\begin{equation}
\ket{\tilde{S}_N} = \frac{1}{\left(\frac{N}{2}\right)! \sqrt{\frac{N}{2} + 1}}
\sum_{\pi}
z! \, \left( \frac{N}{2} - z \right)! \, (-1)^{\frac{N}{2} - z}
\, \pi \left[\ket{01}^{\otimes \frac{N}{2}}\right],
\end{equation}
where the sum runs over all permutation $\pi$ and $z$ is the number of zeros in the first $N/2$ position (e.g., $z=2$ in $\ket{010110}$). Note that $\ket{\tilde{S}_2} = \ket{S_2}$, $\ket{\tilde{S}_4} = \ket{S_4^{(2)}}$, and
\begin{align} \nonumber
    \ket{\tilde{S}_6}
    &= \frac{1}{6} \big(3\ket{000111} - \ket{001011} - \ket{001101} - \ket{001110} - \ket{010011} - \ket{010101} - \ket{010110} + \ket{011001} \\ \nonumber
    &\quad
    + \ket{011010} + \ket{011100} - \ket{100011} - \ket{100101} - \ket{100110} + \ket{101001} + \ket{101010} + \ket{101100}
    \\
    &\quad
    + \ket{110001} + \ket{110010} + \ket{110100} - 3\ket{111000} \big).
\end{align}

\section{Additional notes on Observation~\ref{ob:sfractionValue}}\label{ap:additional_notes_sfraction}
Here we provide the explicit form of $s_2$ and $s_3$ in Observation~\ref{ob:sfractionValue} in the main text:
\begin{subequations}
        \begin{align}
            s_2 &=
            \frac{(2 N-3) [N^4 \delta _{{\rm even},N}+(N^2-1)^2 \delta _{{\rm odd},N}]}{8 (N-1)^3 N},
            \\
            s_3 &=
            \frac{\frac{27}{2} N^6}{c_1 + c_2 + 3 N [N (3 N^3 - 6 N^2 - 40 N - 60) - 340]},
            \\
            c_1&= \frac{380 (164-71 N)}{3 (N-5) N+20}
            + \frac{12800 (N-1)}{[3 (N-5) N+20]^2}
            -3084,
            \\ 
            c_2
            &=3 \sqrt{\frac{N^2 [N(N c_3 -1440) + 480]^3}
            {(N-2) (N-1) [3 (N-5) N+20]^4}},
            \\
            c_3&= 
            3 N \{N [3 (N-9) N+128]-360\}+1720.
        \end{align}
\end{subequations}

Also, in Fig.~\ref{fig:7}, we plot the quantity $s_k \equiv s(J_\alpha^k)$ up to $N \leq 9$ and $k \leq 9$, where we numerically compute the separability bounds. Although Observation~\ref{ob:sfractionValue} in the main text disproves the presence of a simple hierarchical order $s_{k} > s_{k+1}$ for a large $N$, we would like to conjecture that another hierarchical order $s_{k} > s_{k+2}$ exists for a large $N$.

\begin{figure}[t]
    \centering
    \includegraphics[width=0.6\columnwidth]{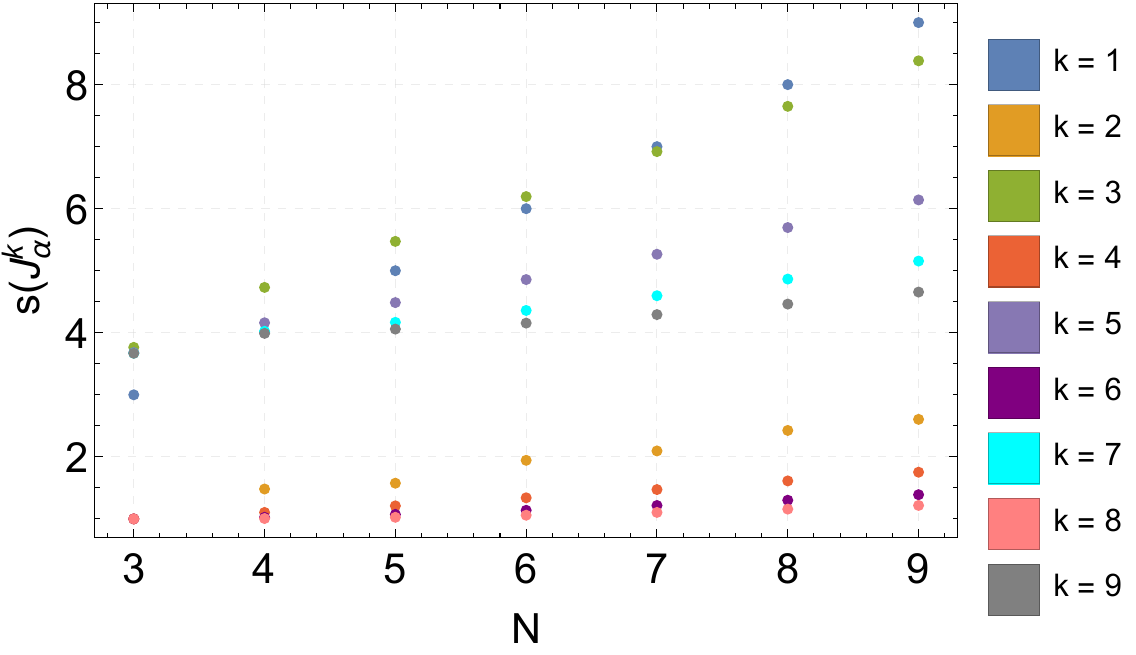}
    \caption{Plot of the quantity $s(J_\alpha^k)$ in Eq.~(\ref{eq:sfraction}) in the main text for $N\in [3,9]$ and $k\in [1,9]$.}
    \label{fig:7}
\end{figure}

\section{Metrologically useful Hamiltonian}\label{ap:complicatedHam}
Here we consider the Hamiltonian $H_{\alpha, \beta} = \mu J_\alpha + \nu J_\beta^2$, for parameters $\mu, \nu$ and different directions $\alpha, \beta$. In general, there is no analytical solution for the eigenproblem of $H_{\alpha, \beta}$, so we should numerically compute $\mathcal{C}_{\rm ent}(H_{\alpha, \beta})$. On the other hand, we have collected numerical evidence that $\mathcal{C}_{\rm sep}(H_{\alpha, \beta})$ can be attained by a symmetric state $\ket{\phi_{\rm sep}^\ast}$. Based on that, in the following, we consider a symmetric state to derive $\mathcal{C}_{\rm sep} (H_{\alpha, \beta})$ via $\max \va{H_{\alpha, \beta}}_{\ket{\phi_{\rm sep}^\ast}}$:
\begin{align} \nonumber
    \va{H_{\alpha, \beta}}_{\ket{\phi_{\rm sep}^\ast}}
    &= \ex{H_{\alpha, \beta}^2}_{\ket{\phi_{\rm sep}^\ast}}
    - \ex{H_{\alpha, \beta}}_{\ket{\phi_{\rm sep}^\ast}}^2
    \\
    &=\frac{N}{8} \left\{2 \left(1-\alpha^2\right) \mu^2-4 \alpha  \beta^2 \mu  \nu  (N-1)
    +(1-\beta) (1+\beta) \nu^2 (N-1) \left[\beta^2 (2 N-3)+1\right]\right\},
    \label{eq:symmetricHab}
\end{align}
where we denote that $\alpha = \braket{\phi|\sigma_\alpha|\phi}$ and $\beta = \braket{\phi|\sigma_\beta|\phi}$ and used
\begin{align} \nonumber
    H_{\alpha, \beta}^2
    &= \mu^2 J_\alpha^2 + \nu^2 J_\beta^4
    + \mu \nu (J_\beta^2 J_\alpha + J_\alpha J_\beta^2)
    \\ \nonumber
    &= 
    \frac{\mu^2}{4}\Biggl\{
    N\eins + \sum_{i\neq j}^N \sigma_\alpha^{(i)} \sigma_\alpha^{(j)}
    \Biggr\}
    + \frac{\nu^2}{16}\Biggl\{
    (3N-2)N \eins
    + 2(3N-4) \sum_{i\neq j}^N \sigma_\beta^{(i)} \sigma_\beta^{(j)}
    + \sum_{i\neq j \neq k \neq l}^N
    \sigma_\beta^{(i)}\sigma_\beta^{(j)}\sigma_\beta^{(k)}\sigma_\beta^{(l)}
    \Biggr\}
    \\
    &\quad    
    + \frac{\mu \nu }{8}
    \Biggl\{
    2N\sum_{i=1}^N \sigma_\alpha^{(i)}
    + \sum_{i\neq j \neq k}^N
    \left[
    \sigma_\beta^{(i)} \sigma_\beta^{(j)} \sigma_\alpha^{(k)}
    + \sigma_\alpha^{(i)} \sigma_\beta^{(j)} \sigma_\beta^{(k)}
    \right]    
    \Biggr\}.
\end{align}
Here $\alpha \in [-1,1]$ and $\beta \in [-1,1]$ obeys the purity condition $\alpha^2 + \beta^2 \leq 1$. Performing the maximization over $\alpha, \beta$ with these conditions, we can find the separability bound $\mathcal{C}_{\rm sep}(H_{\alpha, \beta})$. Note that this optimization includes four parameters, so the explicit forms of the separability bound may not be available. In Fig.~\ref{fig:3}, we illustrate the comparison of the separability bound $\mathcal{C}_{\rm sep}(H_{\alpha, \beta})$ in Eq.~(\ref{eq:SQL}) in the main text and the quantity $s(H_{\alpha, \beta})$ in Eq.~(\ref{eq:sfraction}) in the main text between different values of $\mu$ (setting $\nu = 1-\mu$).

\begin{figure}[t]
    \centering
    \includegraphics[width=1.0\columnwidth]{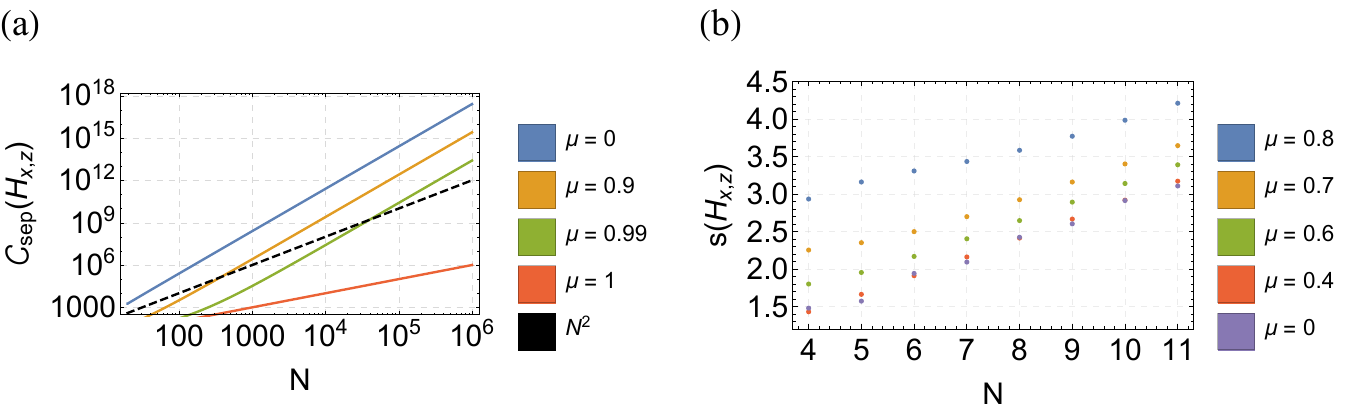}
    \caption{Comparison of the separability bound $\mathcal{C}_{\rm sep}(H_{\alpha, \beta})$ in Eq.~(\ref{eq:SQL}) and the quantity $s(H_{\alpha, \beta})$ in Eq.~(\ref{eq:sfraction}) in the main text for Hamiltonian $H_{\alpha, \beta} = \mu J_\alpha + \nu J_\beta^2$.
    Here we set $\alpha = x, \beta = z$, $\nu = 1-\mu$.
    (a) $\mu  = 0, 0.9, 0.99, 1$, and $20 \leq N \leq 10^6$.
    (b) $\mu  = 0, 0.4, 0.6, 0.7, 0.8$, and $4 \leq N \leq 11$.}
    \label{fig:3}
\end{figure}

\section{Spin squeezing}\label{ap:spinsqueezing}
Here we first explain the basic notions of spin squeezing. Then we show detailed calculations and discuss the relation with metrologically useful entanglement.
\begin{itemize}
    \item 
    The concept of spin-squeezed states was initially introduced to achieve higher accuracy compared to classical interferometers~\cite{wineland1992spin} (also see \cite{ma2011quantum}). Several previous works~\cite{wang2003spin,korbicz2005spin,toth2007optimal,toth2009spin} have considered an $N$-qubit state $\vr$ as spin-squeezed if its entanglement can be detected from $\ex{J_\alpha}_\vr$ and $\ex{J_\alpha^2}_\vr$. Refs.~\cite{toth2007optimal,toth2009spin} proposed the optimal spin-squeezing inequalities with optimal measurement directions: any $N$-qubit fully separable state obeys
\begin{subequations}
    \begin{align} \label{eq:ss1}
    \tr(\Gamma) &\geq \frac{N}{2},
    \\
    \chi_{\rm min}(\mathfrak{X}) &\geq \tr(C) - \frac{N}{2},
    \label{eq:ss2}
    \\
    \chi_{\rm max}(\mathfrak{X}) &\leq  (N-1) \tr(\Gamma) - \frac{N(N-2)}{4},
    \label{eq:ss3}
\end{align}
\end{subequations}
where $C$ denotes the matrix with $C_{kl} = \ex{J_k J_l + J_k J_l}/2$, $\Gamma$ denotes the matrix with $\Gamma_{kl} = C_{kl} - \ex{J_k} \ex{J_l}$ for $k,l=x,y,z$, $\mathfrak{X} = (N-1)\Gamma + C$, and $\chi_{\rm max/min}$ are respectively the largest/smallest eigenvalues of matrix $\mathfrak{X}$. Violation of these inequalities implies that the state is multipartite spin-squeezed entangled. The W state, Dicke state, and singlet state $\ket{S_N}$ in the main text are detected by these bounds, while the GHZ state cannot be detected. We remark that the previous bound $\mathcal{C}_{\rm sep}(J_\alpha) = N$, initially proposed Ref.~\cite{pezze2009entanglement}, can detect the entangled states that cannot be detected by these spin-squeezing inequalities in Eqs.~(\ref{eq:ss1}),~(\ref{eq:ss2}),~and~(\ref{eq:ss3}), such as the GHZ state.

\item 
    In the following, we show the calculation of Eqs.~(\ref{eq:ss1}),~(\ref{eq:ss2}),~and~(\ref{eq:ss3}) for the mixed state $\vr_{\eta}$ in Eq.~(\ref{eq:testmix}):
\begin{equation}
    \vr_{\eta}
    = \eta \ket{\Phi}\!\bra{\Phi} + \frac{1-\eta}{2^N}\eins,
    \quad
    \ket{\Phi}
    = \sqrt{\lambda_1} \ket{0}^{\otimes N}
    +\sqrt{\lambda_2} \ket{1}^{\otimes N}
    +\sqrt{\lambda_3} \ket{S_N},
\end{equation}
where, for the sake of simplicity, we set $\alpha = z$. For the pure case $\eta = 1$, we summarize:
\begin{subequations}
\begin{align}
    &\ex{J_x}_{\Phi} = 0,
    \quad
    \ex{J_y}_{\Phi} = 0,
    \quad
    \ex{J_z}_{\Phi} = \tfrac{N}{2}(\lambda_1 - \lambda_2),
    \\
    &\ex{J_x^2}_{\Phi} = \tfrac{N}{4}(\lambda_1 + \lambda_2),
    \quad
    \ex{J_y^2}_{\Phi} = \tfrac{N}{4}(\lambda_1 + \lambda_2),
    \quad
    \ex{J_z^2}_{\Phi} = \left(\tfrac{N}{2}\right)^2(\lambda_1 + \lambda_2),
    \\
    &\ex{J_x J_y + J_y J_x} = 0,    
    \quad
    \ex{J_y J_z + J_z J_y} = 0,
    \quad
    \ex{J_z J_x + J_x J_z} = 0,
    \\
    &C = (\lambda_1 + \lambda_2) \, {\rm diag}\left(\tfrac{N}{4}, \tfrac{N}{4}, \tfrac{N^2}{4}\right),
    \quad
    \Gamma = {\rm diag} \left\{ \tfrac{N}{4}(\lambda_1 + \lambda_2), \tfrac{N}{4}(\lambda_1 + \lambda_2), \tfrac{N^2}{4} [\lambda_1 + \lambda_2 - (\lambda_1 - \lambda_2)^2]\right\}.
\end{align}
\end{subequations}
This can directly lead to the general case ($0 \leq \eta \leq 1$):
\begin{subequations}
\begin{align}
    C &= {\rm diag}\left[
    f(N, \eta, \lambda_1, \lambda_2),
    f(N, \eta, \lambda_1, \lambda_2),
    \tfrac{N^2}{4} \eta (\lambda_1 + \lambda_2) + \tfrac{N}{4}(1-\eta)
    \right],
    \\
    \Gamma &= {\rm diag}\left\{
    f(N, \eta, \lambda_1, \lambda_2),
    f(N, \eta, \lambda_1, \lambda_2),
    \tfrac{N^2}{4} [\eta (\lambda_1 + \lambda_2) - \eta^2(\lambda_1 - \lambda_2)^2] + \tfrac{N}{4}(1-\eta)
    \right\},
\end{align}
\end{subequations}
where $f(N, \eta, \lambda_1, \lambda_2) = \tfrac{N}{4} \eta(\lambda_1 + \lambda_2) + \tfrac{N}{4}(1-\eta)$.

    \item 
    For instance, the optimal entangled state $\ket{\Phi(\lambda)} = \sqrt{\lambda}\ket{0}^{\otimes N}+\sqrt{{1}/{2} - \lambda}\ket{1}^{\otimes N} + \sqrt{{1}/{2}} \ket{S_N}$ for $J_z^k$ with even $k$, mentioned in the main text, violates the inequality in Eq.~(\ref{eq:ss3}) for any $\lambda \in [0,1/2]$ and any $N \geq 4$. That is, state $\ket{\Phi(\lambda)}$ is spin-squeezed. Also, the mixed state $\vr_{\eta,\lambda} = \eta \ket{\Phi(\lambda)}\!\bra{\Phi(\lambda)} + \frac{1-\eta}{2^N}\eins$, violates the inequality in Eq.~(\ref{eq:ss3}) when $\eta > 2 (N-1)/(3N-4)$.
    
    \item 
    In Fig~\ref{fig:4}, we consider the case of $N=10$ and illustrate the classification of metrologically useful entanglement in the space with $F_Q(\vr_{\eta}, J_\alpha)$ and $F_Q(\vr_{\eta}, J_\alpha^2)$, using the previous bound $\mathcal{C}_{\rm sep}(J_\alpha) = N$~\cite{pezze2009entanglement}, our bound $\mathcal{C}_{\rm sep}(J_\alpha^2)$ in Eq.~(\ref{eq:sepbound2}), and the optimal spin-squeezing inequalities~\cite{toth2007optimal,toth2009spin}. In Fig.~\ref{fig:5}, we illustrate the comparison between our bounds presented in Observation~\ref{ob:sepbounds} in the main text, the previous bound $\mathcal{C}_{\rm sep}(J_\alpha) = N$, and the optimal spin-squeezing inequalities in Eqs.~(\ref{eq:ss1}),~(\ref{eq:ss2}),~and~(\ref{eq:ss3}), for the mixed state $\vr_{\eta}$ in Eq.~(\ref{eq:testmix}) in the main text (also see below). Our bounds can also discover the metrologically useful entangled states that cannot be detected by not only the spin-squeezing inequalities but also the previous bound $\mathcal{C}_{\rm sep}(J_\alpha) = N$.
\end{itemize}
 
\begin{figure}[t]
    \centering
    \includegraphics[width=1.0\columnwidth]{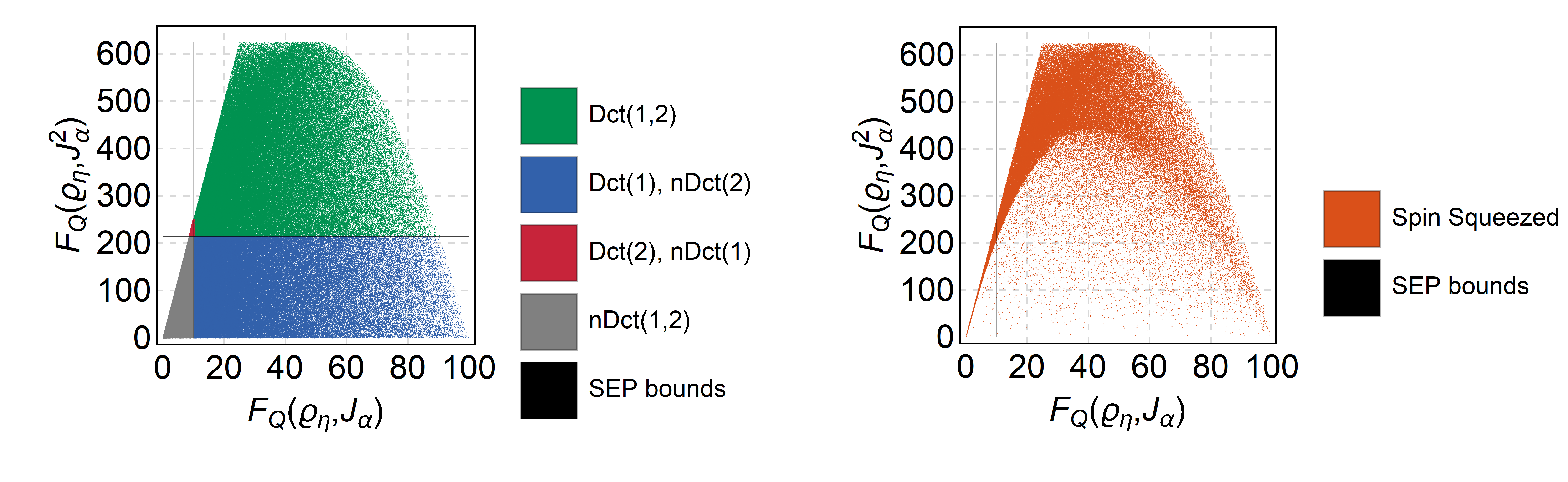}
    \caption{Classification of metrologically useful entanglement for the mixed state $\vr_{\eta}$ with $N=10$, given in Eq.~(\ref{eq:testmix}) in the main text, in the space coordinated by $F_Q(\vr_{\eta}, J_\alpha)$ and $F_Q(\vr_{\eta}, J_\alpha^2)$. These figures are created from a random sample of $10^6$ points $\lambda_1, \lambda_2, \eta \in [0,1]$ with~$0 \leq \lambda_1 + \lambda_2 \leq 1$ in Eq.~(\ref{eq:maxusefulstate}) in the main text.
    Left: Dct($k$)/nDct($k$) for $k=1,2$ means the areas that can/cannot be detected by SEP bounds given in~$\mathcal{C}_{\rm sep}(J_\alpha) = N$~\cite{pezze2009entanglement} and $\mathcal{C}_{\rm sep}(J_\alpha^2)$ presented in Observation~\ref{ob:sepbounds} in the main text.
    Right: Spin Squeezed means the areas that can be detected by all the optimal spin-squeezing inequalities in Eqs.~(\ref{eq:ss1}),~(\ref{eq:ss2}),~and~(\ref{eq:ss3}), presented in Refs.~\cite{toth2007optimal,toth2009spin}.
    }
    \label{fig:4}
\end{figure}

\begin{figure}[t]
    \centering
    \includegraphics[width=1.0\columnwidth]{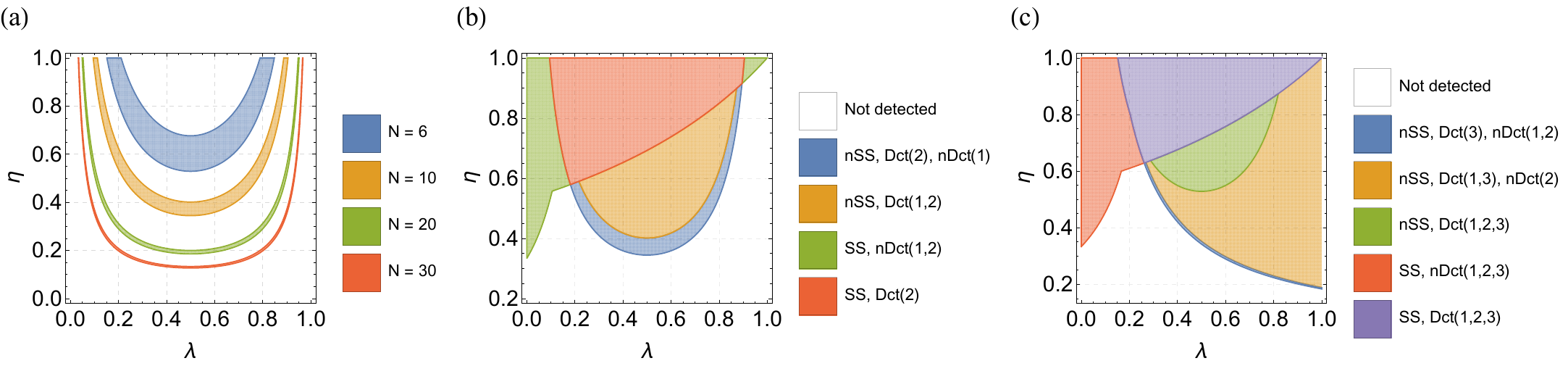}
    \caption{
    Entanglement criteria for the mixed state $\vr_{\eta}$, given in Eq.~(\ref{eq:testmix}) in the main text, in the $\lambda-\eta$ plane. (a) $N=6,10,20,30$, and $\lambda_1 = \lambda$ and $\lambda_2 = 0$. Each colored area illustrates the entangled states that cannot be detected by the previous bound $\mathcal{C}_{\rm sep}(J_\alpha) = N$~\cite{pezze2009entanglement}, but can be detected by Eq.~(\ref{eq:sepbound2}) presented in Observation~\ref{ob:sepbounds} in the main text, thus marking the improvement of this manuscript compared with the previous result. (b) $N=10$, $\lambda_1 = \lambda$, and $\lambda_2 = 0$. (c) $N=6$ and $\lambda_1 = \lambda_2 = \lambda/2$. Here, Dct($k$)/nDct($k$) means the areas that can/cannot be detected by the bounds with the QFI for $k=1,2,3$, respectively given in $\mathcal{C}_{\rm sep}(J_\alpha) = N$~\cite{pezze2009entanglement} and Eqs.~(\ref{eq:sepbound2})~and~(\ref{eq:sepbound3}) presented in Observation~\ref{ob:sepbounds} in the main text. Also, SS/nSS means the areas that can/cannot be detected by the optimal spin-squeezing inequalities in Eqs.~(\ref{eq:ss1}),~(\ref{eq:ss2}),~and~(\ref{eq:ss3}), presented in Refs.~\cite{toth2007optimal,toth2009spin}.
    }
    \label{fig:5}
\end{figure}

\section{Average QFI of random pure symmetric states}\label{ap:averageQFI}
Here we first explain the basics of Haar random integration and then describe the average QFI of random pure states of the symmetric subspace. Finally, we give the proof of Observation~\ref{ob:tfractionaverageValue} in the main text.
\begin{itemize}
    \item 
    First, we give basic notions of Haar unitaries, following the descriptions of Ref.~\cite{cieslinski2023analysing} (for further details, also see Refs.~\cite{harrow2013church,oszmaniec2016random,roberts2017chaos,mele2024introduction}). Let $f(U)$ be a function on a unitary $U$. Consider an integral of $f(U)$ over the unitary group concerning the Haar measure. Importantly, the Haar unitary integral is invariant under left and right shifts via multiplication by a unitary $V$, i.e., 
     \begin{equation}\label{eq:invarianceleftright}
         \int dU \, f(U) = \int dU \, f(VU) = \int dU \, f(UV).
     \end{equation}
     For a $t$-particle and $D$-dimensional operator $X$, let us consider
     \begin{equation}
        \Lambda_t (X) = \int dU \, U^{\otimes t} X (U^\dagger)^{\otimes t}.
     \end{equation}
     Due to the invariance property in Eq.~(\ref{eq:invarianceleftright}), one can show that $\Lambda_t (X)$ commutes with all unitaries $V^{\otimes t}$. Using this property and the Schur-Weyl duality, the operator $\Lambda_t (X)$ can be written in a linear combination of permutation operators $W_\pi$ with a permutation $\pi$ on the symmetric group ${\rm Sym}(t)$ of degree $t$, i.e.,
     \begin{equation}
         \Lambda_t (X) = \sum_{\pi \in {\rm Sym}(t)} x_{\pi} W_\pi,
     \end{equation}
     where each of $x_{\pi}$ can be found with the help of the so-called Weingarten calculus. In particular, in the case of $X = \ket{0}\! \bra{0}^{\otimes t}$ with a $D$-dimensional pure state $\ket{0}$, the operator $\Lambda_t (X)$ can be seen as a uniform random ensemble overall pure state. Letting $\ket{\psi}= U\ket{0}$ be a pure quantum state, one can write that
     \begin{equation}\label{eq:haarstateintegrals}
         \Lambda_t (\ket{0}\! \bra{0}^{\otimes t}) 
         = \int d\psi \, \ket{\psi}\! \bra{\psi}^{\otimes t}
         = \frac{\mathcal{P}_{\rm sym}^{(t, D)}}{d_{\rm sym}^{(t, D)}},
     \end{equation}
     where $\mathcal{P}_{\rm sym}^{(t, D)} = ({1}/{t!}) \sum_{\pi \in {\rm Sym}(t)} W_\pi$ is the projector onto the permutation symmetric subspace and $d_{\rm sym}^{(t, D)} = \binom{D+t-1}{t}$ is its dimension. Examples for $t=1,2$ are given by
     \begin{equation}\label{eq:formlasintrandpmstate}
         \int d\psi \, \ket{\psi}\! \bra{\psi}
         = \frac{\eins_D}{D},
         \qquad
         \int d\psi \, \ket{\psi}\! \bra{\psi}^{\otimes 2}
         = \frac{1}{D(D+1)}(\eins_D \otimes \eins_D + \swap),
     \end{equation}
     where $\swap$ denotes the SWAP operator acting on the $D \otimes D$-dimensional space such that $\swap \ket{x} \otimes \ket{y} = \ket{y} \otimes \ket{x}$. Here the Haar integrals in Eq.~(\ref{eq:formlasintrandpmstate}) are taken over the whole Hilbert space.
     
    \item 
    Let $\ket{\psi_{\rm sym}}$ be a $N$-particle and $d$-dimensional pure symmetric state. In Ref.~\cite{oszmaniec2016random}, the Haar integral over the symmetric subspace was discussed. More specifically, the following formulas were introduced:
    \begin{equation}\label{eq:formlasintrandpmstatesym}
         \int d\psi_{\rm sym} \, \ket{\psi_{\rm sym}}\! \bra{\psi_{\rm sym}}
         = \frac{\mathcal{P}_{\rm sym}^{(N, d)} }{d_{\rm sym}^{(N, d)}},
         \quad
         \int d\psi \, \ket{\psi_{\rm sym}}\! \bra{\psi_{\rm sym}}^{\otimes 2}
         = \frac{\mathcal{P}_{\rm sym}^{(N, d)} \otimes \mathcal{P}_{\rm sym}^{(N, d)} + \swap}{d_{\rm sym}^{(N, d)}\left(d_{\rm sym}^{(N, d)}+1\right)},
     \end{equation}
      where $\mathcal{P}_{\rm sym}^{(N, d)}$ is the projector onto the permutation symmetric subspace with the dimension $d_{\rm sym}^{(N, d)} = \binom{d+N-1}{N}$ and $\swap$ is the SWAP operator acting on the $d_{\rm sym}^{(N, d)} \otimes d_{\rm sym}^{(N, d)}$-dimensional symmetric subspace. Ref.~\cite{oszmaniec2016random} has used the formulas in Eq.~(\ref{eq:formlasintrandpmstatesym}) to compute the average QFI of random pure symmetric states for the linear Hamiltonian $H_{\rm L} = \sum_{i=1}^N H_i$:
      \begin{equation}\label{eq:averageQFI}
          \Bar{\mathcal{F}}_Q(H_{\rm L})
          \equiv \int d\psi_{\rm sym} \, F_Q(\ket{\psi_{\rm sym}}, H_{\rm L}).
      \end{equation}
      In particular, in the case of $N$-qudit ($d=2$) and $H_{\rm L} = J_\alpha$, the average QFI in Eq.~(\ref{eq:averageQFI}) is expressed as
      \begin{equation}
          \Bar{\mathcal{F}}_Q(J_\alpha)
          = \frac{N(N+1)}{3}.
      \end{equation}
      For more general expressions (linear Hamiltonians in higher dimensions), see Lemma 8 in Ref.~\cite{oszmaniec2016random}.
\end{itemize}

In the following, we give the proof of Observation~\ref{ob:tfractionaverageValue} in the main texts:
\begin{observation}
Consider an $N$-qubit Hamiltonian $H_{\rm NL} = J_\alpha^k$. For large $N$, we obtain that $\Bar{\mathcal{F}}_Q(J_\alpha^k) \propto ({1}/{k}) \mathcal{C}_{\rm ent}(J_\alpha^k)$ (with analytical expression valid for any $k$ and $N$ reported in the proof), leading to the scaling behavior $\mathcal{O}(N^{2k}/k)$. Also, defining the quantity $t_k \equiv \Bar{\mathcal{F}}_Q(J_\alpha^k)/\mathcal{C}_{\rm sep}(J_\alpha^k)$, we have that $t_2 < t_3 < t_1$ for $N \geq 4$. Here, $t_1 = (N+1)/3$ and the explicit forms of $t_2, t_3$ are given by
\begin{subequations}
        \begin{align}
            t_2 &=
            \frac{2(N+1)(N+3)(2N-3)}{45(N-1)^2},
            \\
            t_3 &= \frac{9}{14}\frac{N (1 + N) [16 +  3 N (-8 + N^3+ 4N^2)]}
            {3 N [-340 + N (-60 + N (-40 + 3N^2 -6N )] + c_1 + c_2},
            \\
            c_1&= \frac{380 (164-71 N)}{3 (N-5) N+20}
            + \frac{12800 (N-1)}{[3 (N-5) N+20]^2}
            -3084,
            \\ 
            c_2
            &=3 \sqrt{\frac{N^2 [N(N c_3 -1440) + 480]^3}
            {(N-2) (N-1) [3 (N-5) N+20]^4}},
            \\
            c_3&= 
            3 N \{N [3 (N-9) N+128]-360\}+1720.
        \end{align}
\end{subequations}
\end{observation}

\begin{proof}
We begin by writing
    \begin{subequations}
    \begin{align}
          \Bar{\mathcal{F}}_Q(J_\alpha^k)
          &= \int d\psi_{\rm sym} \, F_Q(\ket{\psi_{\rm sym}}, J_\alpha^k)
          \\ \label{eq:variancefisher}
          &= 4\int d\psi_{\rm sym} \,
          \tr \left[\ket{\psi_{\rm sym}}\! \bra{\psi_{\rm sym}} J_\alpha^{2k} \right]
          - 4\int d\psi_{\rm sym} \,
          \tr \left[\ket{\psi_{\rm sym}}\! \bra{\psi_{\rm sym}} J_\alpha^{k} \right]^2
          \\
          &= 
          4\tr \left[\int d\psi_{\rm sym} \, \ket{\psi_{\rm sym}}\! \bra{\psi_{\rm sym}} J_\alpha^{2k} \right]
          - 4\tr \left[ \int d\psi_{\rm sym} \, \ket{\psi_{\rm sym}}\! \bra{\psi_{\rm sym}}^{\otimes 2} J_\alpha^{k} \otimes J_\alpha^{k} \right]
          \\ \label{eq:formulassymmrandom}
          &= \frac{4}{d_{\rm sym}^{(N, 2)}}
          \tr \left[\mathcal{P}_{\rm sym}^{(N, 2)}  J_\alpha^{2k} \right]
          - \frac{4}{d_{\rm sym}^{(N, 2)}\left(d_{\rm sym}^{(N, 2)}+1\right)}
          \tr \left[ \left(\mathcal{P}_{\rm sym}^{(N, 2)} \otimes \mathcal{P}_{\rm sym}^{(N, 2)} + \swap\right)
          J_\alpha^{k} \otimes J_\alpha^{k} \right]
          \\ \label{eq:somestep1}
          &= \frac{4}{N + 1} \left[
          \tau_{N, 2k}
          - \frac{1}{N + 2}
          \left( \tau_{N, k}^2 + \tau_{N, 2k}
          \right)
          \right].
    \end{align}
    \end{subequations}
    In Eq.~(\ref{eq:variancefisher}), we inserted the definition of the variance. In Eq.~(\ref{eq:formulassymmrandom}), we used the formilas in Eq.~(\ref{eq:formlasintrandpmstatesym}). In Eq.~(\ref{eq:somestep1}), we used that $d_{\rm sym}^{(N, 2)} = N + 1$, denoted that $\tau_{N, k} \equiv \tr [\mathcal{P}_{\rm sym}^{(N, 2)}  J_\alpha^{k} ]$, and employed the SWAP trick: $\tr[\swap X \otimes Y] = \tr[XY]$ for operators $X,Y$. Then we only need to evaluate the term $\tau_{N, k}$.
    
    To proceed, let us recall that the $N$-qubit symmetric subspace is spanned by the $N$-qubit Dicke state with $m$ excitations $\{\ket{D_{N,m}}\}_{m=0}^N$ given by
    \begin{equation}
        \ket{D_{N,m}} = \binom{N}{m}^{-\frac{1}{2}}
        \sum_{k} \pi_k
        \left[\ket{1}^{\otimes m} \otimes \ket{0}^{\otimes (N-m)}
        \right],
    \end{equation}
    where the summation in $\sum_{k} \pi_k$ is over all permutations between the qubits that lead to different terms. Then the projector onto the symmetric subspace $\mathcal{P}_{\rm sym}^{(N, 2)}$ is written as
    \begin{equation}\label{eq:projectordicke}
        \mathcal{P}_{\rm sym}^{(N, 2)}
        = \sum_{m=0}^N \ket{D_{N,m}} \! \bra{D_{N,m}}.
    \end{equation}
    Inserting Eq.~(\ref{eq:projectordicke}) into $\tau_{N, k} = \tr [\mathcal{P}_{\rm sym}^{(N, 2)}  J_\alpha^{k} ]$, we obtain
    \begin{equation}\label{eq:tauformula}
        \tau_{N, k}
        = \sum_{m=0}^N 
        \tr \left[\ket{D_{N,m}} \! \bra{D_{N,m}}  V^{\otimes N} J_z^{k} (V^\dagger)^{\otimes N} \right]
        = \sum_{m=0}^N 
        \tr \left[\ket{D_{N,m}} \! \bra{D_{N,m}}  J_z^{k} \right]
        = \sum_{m=0}^N  \left(\frac{N}{2}-m\right)^k,
    \end{equation}
    where we first used that any $J_\alpha$ can be expressed using $J_z$ and a collective local unitary $V^{\otimes N}$, i.e.,  $J_\alpha = V^{\otimes N} J_z (V^\dagger)^{\otimes N}$, then the invariant property discussed in Eq.~(\ref{eq:haarstateintegrals}), and finally employed the fact that the Dicke state $\ket{D_{N,m}}$ is the eigenstate of $J_z$ with the eigenvalues $\tfrac{N}{2}-m$, i.e., $J_z \ket{D_{N,m}} = (\tfrac{N}{2}-m)\ket{D_{N,m}}$.
    
    One can further expand $\tau_{N,k}$ in Eq.~(\ref{eq:tauformula}) as
    \begin{subequations}
    \begin{align}
        \tau_{N, k}
        &= \left(\frac{N}{2}\right)^k + \sum_{m=1}^N \sum_{p=0}^k \binom{k}{p} \left(\frac{N}{2}\right)^{k-p} (-1)^p m^p
        \\ \label{eq:Faulhaberfurmula}
        &= \left(\frac{N}{2}\right)^k + \sum_{p=0}^k \binom{k}{p} \left(\frac{N}{2}\right)^{k-p} (-1)^p \frac{1}{p+1} \sum_{r=0}^p \binom{p+1}{r} B_r N^{p-r+1},
    \end{align}
    \end{subequations}
    where we used Faulhaber's formula
    \begin{equation}
        \sum_{m=1}^N m^p = \frac{1}{p+1} \sum_{r=0}^p \binom{p+1}{r} B_r N^{p-r+1}.
    \end{equation}
    Here $B_r$ denotes the Bernoulli numbers, e.g.,
    \begin{equation}
        B_0 = 1,
        \quad B_1 = \frac{1}{2},
        \quad B_2 = \frac{1}{6},
        \quad B_3 = 0,
        \quad B_4 = -\frac{1}{30},
        \quad B_5 = 0,
        \quad B_6 = \frac{1}{42}.
    \end{equation} 
    These lead yields
    \begin{subequations}
    \begin{align}
        \tau_{N, 0} &= N+1,
        \ \tau_{N, 1} = 0,
        \ \tau_{N, 2} = \frac{N(N+1)(N+2)}{12},
        \ \tau_{N, 3} = 0,
        \\
        \tau_{N, 4} &= \frac{N(N+1)(N+2)[3N(N+2)-4]}{240},
        \ \tau_{N, 5} = 0,
        \ \tau_{N, 6} = \frac{N(N+1)(N+2)[3 N (N^3 + 4N^2 -8)+16]}{1344}.
    \end{align}
    \end{subequations}
    Inserting these values into Eq.~(\ref{eq:somestep1}), we thus obtain
    \begin{align}
        \Bar{\mathcal{F}}_Q(J_\alpha)
        &=  \frac{N(N+1)}{3},
        \\
        \Bar{\mathcal{F}}_Q(J_\alpha^2)
        &=  \frac{N(N-1)(N+1)(N+3)}{45},
        \\
        \Bar{\mathcal{F}}_Q(J_\alpha^3)
        &=  \frac{N(N+1)[3 N (N^3 + 4N^2 -8)+16]}{336}.
    \end{align}
    We note that one can similarly find the higher-order expression of $\Bar{\mathcal{F}}_Q(J_\alpha^k)$ for $k \geq 4$. 

    Let us discuss the scaling behavior for large $N$. Denoting $x = \tfrac{N}{2}-m$ and then replacing the sum in $\tau_{N, k}$ in Eq.~(\ref{eq:tauformula}) by the integral for large $N$, we find
    \begin{equation}
        \tau_{N, k}
        = \sum_{x =-\frac{N}{2}}^{\frac{N}{2}} x^k
        \approx
        \int_{-\frac{N}{2}}^{\frac{N}{2}} dx \, x^k
        = \frac{1+ (-1)^{k}}{k+1} \left( \frac{N}{2} \right)^{k+1}.
    \end{equation}
    Inserting this into Eq.~(\ref{eq:somestep1}), we thus obtain
    \begin{align}
        \Bar{\mathcal{F}}_Q(J_\alpha^k)
        \approx
        \frac{4}{N} \left[
        \frac{N^{2k+1}}{2k \cdot 4^k}
          - \frac{N^{2k+2}[1+ (-1)^{k}]^2}{N \cdot k^2 \cdot 4^{k+1}}
          \right]
        =\frac{N^{2k}}{4^{k-1}} \frac{1}{2k} \delta_{k, \text{odd}}
        + \frac{N^{2k}}{4^k} \frac{2(k-2)}{k^2} \delta_{k, \text{even}}.
    \end{align}
    Recalling $\mathcal{C}_{\rm ent}(J_\alpha^k)$ in Eq.~(\ref{eq:allbound}) in the main text, we  write
    \begin{equation}
    \mathcal{C}_{\rm ent}(J_\alpha^k)
    \approx \frac{N^{2k}}{4^{k-1}} \delta_{k, \text{odd}}
    + \frac{N^{2k}}{4^{k}} \delta_{k, \text{even}},
    \quad
    {\rm for \ large \ } N.
    \end{equation}
    Then we can arrive at
    $\Bar{\mathcal{F}}_Q(J_\alpha^k) \propto (1/k) \mathcal{C}_{\rm ent}(J_\alpha^k)$ for large $N$. Hence, we can complete the proof.
\end{proof}

\noindent
\textbf{Remark:}
In Ref.~\cite{oszmaniec2016random}, it was already shown that the typical value of the QFI of random pure symmetric states can be close to the average QFI in Eq.~(\ref{eq:averageQFI}) for the linear Hamiltonians $H_{\rm L}$. More specifically, in Appendix~\ref{ap:complicatedHam} in Ref.~\cite{oszmaniec2016random}, the technical discussion about the concentration of measure inequality~\cite{anderson2010introduction} was described. In particular, in the case of $N$-qubit systems, the following large deviation bound was shown to hold for any $\epsilon \geq 0$: 
\begin{equation}\label{eq:concentrationQFI}
    {\rm Prob} \left[ 
    F_Q(\ket{\psi_{\rm sym}}, H) 
    \leq \Bar{\mathcal{F}}_Q(H) - \epsilon \right]
    \leq \exp \left[-\frac{(N+1) \epsilon^2}{4L_H^2} \right],
\end{equation}
where $L_H$ can be taken as $L_H = 32 \norm{H}^2$ and $\norm{X}$ denotes the operator norm of $X$. For $H = J_\alpha^k$, one has
\begin{equation} \label{eq:deviationbound}
    {\rm Prob} \left[
    F_Q(\ket{\psi_{\rm sym}}, J_\alpha^k)
    \leq \Bar{\mathcal{F}}_Q(J_\alpha^k) - \epsilon \right]
    \leq \exp \left[-\frac{(N+1) \epsilon^2}{4096 (N/2)^{4k}} \right],
\end{equation}
where we used $L_{J_\alpha^k} = 32 (N/2)^{2k}$. Setting $\epsilon =  \Bar{\mathcal{F}}_Q(J_\alpha^k) - \mathcal{C}_{\rm sep}(J_\alpha^k)$, Eq.~(\ref{eq:deviationbound}) implies that the probability that the QFI of a random pure symmetric state $\ket{\psi_{\rm sym}}$ cannot violate the separability bounds $\mathcal{C}_{\rm sep}(J_\alpha^k)$ can decrease as large $N$. In Fig.~\ref{fig:6}, we plot the confidence $\gamma \equiv 1 - {\rm Prob} \left[F_Q(\ket{\psi_{\rm sym}}, J_\alpha^k) \leq \Bar{\mathcal{F}}_Q(J_\alpha^k) - \epsilon \right]$ for the certification of the violation of the separability bound $\mathcal{C}_{\rm sep}(J_\alpha^k)$.

\begin{figure}[t]
    \centering
    \includegraphics[width=0.5\columnwidth]{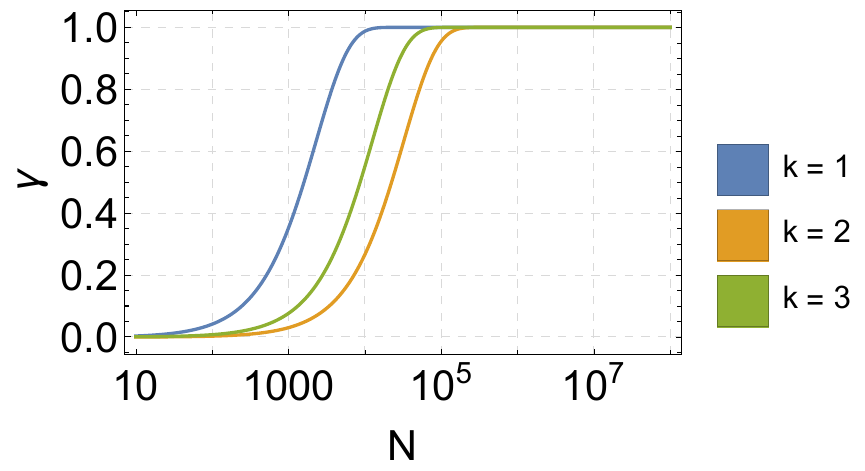}
    \caption{Linear-Log plot of the confidence $\gamma$ computed via Eq.~(\ref{eq:deviationbound}) to certify the violation of the separability bounds $\mathcal{C}_{\rm sep}(J_\alpha^k)$ for $k=1,2,3$ and $10\leq N \leq 10^8$.}
    \label{fig:6}
\end{figure}

%

\end{document}